\documentclass[sigconf]{acmart}

\usepackage{balance} %
\usepackage{csquotes}
\usepackage{mdframed}
\usepackage{lipsum}
\usepackage[normalsize]{subfigure}
\usepackage{wrapfig}
\usepackage[ruled, vlined, commentsnumbered, linesnumbered]{algorithm2e}

\usepackage{xcolor}
\usepackage{lipsum}
\usepackage{amsmath}
\usepackage{tcolorbox}

\usepackage[linesnumbered, ruled]{algorithm2e}
\usepackage{algpseudocode}
\usepackage{relsize} 

\SetCommentSty{mycommfont}
\SetKwInput{KwInput}{Input}                %
\SetKwInput{KwOutput}{Output}              %

\newtheorem{theorem}{\bf{Theorem}}[section]
\newtheorem{lemma}{\bf{Lemma}}[section]

\newtheorem{proposition}[theorem]{\bf{Proposition}}

\newtheorem{definition}[theorem]{\bf{Definition}}
\newmdtheoremenv{policy}[theorem]{\bf{Policy}}

\SetKwRepeat{Do}{do}{while}

\newcommand{\DEL}[1]{}

\newcommand{\squishlist}{
\begin{list}{$\bullet$}
  { \setlength{\itemsep}{0pt}
     \setlength{\parsep}{0pt}
     \setlength{\topsep}{0pt}
     \setlength{\partopsep}{0pt}
     \setlength{\leftmargin}{0em}
     \setlength{\labelwidth}{0em}
     \setlength{\labelsep}{0.2em} } }

\setcopyright{none}
\renewcommand\footnotetextcopyrightpermission[1]{} %

\begin{document}
\title{User Customizable and Robust Geo-Indistinguishability for Location Privacy}

\author{Primal Pappachan}
\affiliation{%
  \institution{Penn State University, USA}
    \country{}
}
\email{primal@psu.edu}

\author{Chenxi Qiu}
\affiliation{%
  \institution{University of North Texas, USA}
      \country{}
}
\email{chenxi.qiu@unt.edu}

\author{Anna Squicciarini}%
\affiliation{%
  \institution{Penn State University, USA}
      \country{}
}
\email{acs20@psu.edu}

\author{Vishnu Sharma Hunsur Manjunath}
\affiliation{%
  \institution{Penn State University, USA}
  \country{}
}
\email{vxh5104@psu.edu}

\newcommand{\pp}[1]{{\color{brown} (Primal: #1)}}
\newcommand{\acs}[1]{{\color{red} (Anna: #1)}}
\newcommand{\cq}[1]{{\color{green} (Chenxi: #1)}}

\newcommand{\tba}[1]{{\color{red} (XXX)}}

\newcommand{\systemName[1]}{{CORGI}}

\newcommand{\vCount}[1]{|{#1}|}

\newcommand{\vLocation}[2]{\textit{s}^{#1}_{#2}}

\newcommand{\vMatrix}[2]{\ensuremath{\mathbf{Z}}^{#1}_{#2}}
\newcommand{\vCell}[2]{\textit{z}^{#1}_{#2}}
\newcommand{\vPruneSet}[1]{\ensuremath{\mathcal{S}_{#1}}}
\newcommand{\prunesize}{\ensuremath{\delta}}

\newcommand{\vTree}[1]{\ensuremath{\mathcal{T}}_{#1}}
\newcommand{\vNodeSet}[1]{\ensuremath{\mathcal{V}}_{#1}}
\newcommand{\vNode}[2]{\textit{v}^{#1}_{#2}}
\newcommand{\vLeafNodeSet}[1]{\ensuremath{\Lambda_{#1}}}
\newcommand{\vLeafNode}[2]{\ensuremath{\lambda}^{#1}_{#2}}
\newcommand{\ordering}{\ensuremath{\prec}}
\newcommand{\children}[1]{\ensuremath{\mathcal{N}\left(#1\right)}}
\newcommand{\vLevel}[2]{\ensuremath{\lambda}^{#1}_{#2}}
\newcommand{\vHeight}[1]{\ensuremath{\mathcal{H}}_{#1}}

\newcommand{\vPrior}[1]{\textit{p}_{#1}}
\newcommand{\vDistance}[2]{\textit{d}_{#1, #2}}

\newcommand{\vPolicySet}[1]{\ensuremath{\mathcal{P}}_{#1}}

\begin{abstract}

Location obfuscation functions generated by existing systems for ensuring location privacy  are monolithic and do not allow users to customize their obfuscation range.
This can lead to the user being mapped in  undesirable locations (e.g., shady neighborhoods)  to the location-requesting services.
Modifying the obfuscation function generated by a centralized server on the user side can result in poor privacy as the original function is not robust against such updates.
Users themselves might find it challenging to  understand the parameters involved in obfuscation mechanisms (e.g., obfuscation range and granularity of location representation) and therefore struggle to set realistic trade-offs between privacy, utility, and customization.
In this paper, we propose a new framework called, \systemName{}, i.e., \underline{C}ust\underline{O}mizable \underline{R}obust \underline{G}eo \underline{I}ndistinguishability, which generates location obfuscation functions that are \emph{robust} against user customization while providing strong privacy guarantees based on the  Geo-Indistinguishability paradigm. 
\systemName{} utilizes a tree representation of a given region to assist users in specifying their privacy and customization requirements.
The server side of \systemName{} takes these requirements as inputs and generates an obfuscation function that satisfies Geo-Indistinguishability requirements and is robust against customization on the user side.
The obfuscation function is returned to the user who can then choose to update the obfuscation function (e.g., obfuscation range, granularity of location representation). 
The experimental results on a real-world dataset demonstrate that \systemName{} can efficiently generate obfuscation matrices that are more robust to the customization by users (e.g., removing 14.28\% of locations only causes 3.07\% Geo-Indistinguishably constraint violations in the matrix generated by \systemName{} compared to 18.58\% Geo-Indistinguishably constraint violations by non-robust approaches). 
\end{abstract}

\maketitle

\pagestyle{plain} %

\section{Introduction}
To date, many location obfuscation mechanisms have been successfully proposed \cite{kim2021survey}. These approaches, often placed in the context of service provisioning,  transform users' actual locations into obfuscated locations to protect their privacy while ensuring the quality of service.
\emph{Geo-Indistinguishability (Geo-Ind)} is one of the most popular privacy criteria used in location obfuscation mechanisms \cite{andres2013geo}. It extends the well known \emph{Differential Privacy (DP)} \cite{dwork2014algorithmic} paradigm to protect location privacy in a rigorous fashion. 
To satisfy Geo-Ind, if two locations are geographically close, their reported obfuscated locations will have similar probability distributions. In other words, it is hard for an adversary to distinguish a true location among nearby ones, given its obfuscation location. 

Despite their potential, there are several limitations with DP-based approaches, such as Geo-Ind, when they are applied to location privacy scenarios\cite{oyaGeoIndLooking}.  
Among these limitations is that the obfuscation functions based on \textit{Geo-Ind} \cite{Chatzikokolakis-PoPETs2015, Xiao-SIGSAC2015} tend to be monolithic as it provides the same obfuscation range and the granularity of location sharing for all users.
The obfuscation range is a set of locations from which an obfuscated location is chosen, and granularity of the location determines the size/semantics of the location being shared (e.g., lat-long pairs, block, county). 
Users may have different privacy needs and utility requirements depending on the context and application scenario. Prior work~\cite{kifer2014pufferfish, he2014blowfish} has looked at customizing the obfuscation range to provide users' customizability based on their own privacy/utility needs. 
However, they focused on statistical releases of data and not point queries that are used for sharing location data.
\cite{cao2020pglp} extended this and applied it to location privacy, where they represented the possible locations of a user and their indistinguishability requirements using nodes and edges in a policy graph.
Their goal is to ensure \textit{Geo-Ind} for any two connected nodes in the graph, and to achieve this, they apply DP-based noise to latitude and longitude independently. 
However, their approach is best suited when locations can be neatly categorized, i.e., indistinguishability among multiple locations in the same category (e.g., restaurants). Also, it does not allow specific customization of an obfuscation function, i.e., remove my home and office from the obfuscation range.

There are several challenges to be addressed in developing such a framework that allows users to customize location obfuscation mechanisms generated by an untrusted server. The generation is done at the server as it is an expensive computation problem that user devices cannot perform.
The first challenge is of specifying the customization parameters. This involves enabling users to easily denote their preferred granularity of location sharing and their preferences for obfuscation range.
Note that the user preferences contain private information and hence could not be directly shared with the server that generates the obfuscation function.
After customization, some of the locations might be removed from the obfuscation range, and with the remaining locations, the obfuscation function has to satisfy strong privacy guarantees.
The second challenge, therefore, pertains to generating an obfuscation function on the server side that is \emph{robust} against any customization on the user side.
The third challenge is doing these operations efficiently, as generating such a customizable obfuscation function is an expensive optimization problem with many constraints.
Efficiency is also a challenge when the user updates their granularity of sharing, and a new obfuscation function has to be generated.

To address the above challenges, we propose \emph{\systemName{}} (\underline{C}ust\underline{O}mizable \underline{R}obust \underline{G}eo \underline{I}ndistinguishability), a framework for generating location obfuscation with strong privacy guarantees (based on \textit{Geo-Ind}) that effectively allows users to balance the trade-off between privacy, utility, and customization.
\systemName{} utilizes an untrusted server for performing the computationally heavy task of generating the obfuscation function while ensuring the privacy of the user.
In order to do so, \systemName{} uses a tree structure which is a semantic representation of a given region that assists users in specifying their customization preferences. These preferences are used to select the obfuscation range and granularity of location sharing and are only selectively shared with the server so as to protect the privacy of the user e.g., only the number of locations to be removed from the obfuscation range and not the exact locations. 
The \systemName{} generates a \emph{robust} obfuscation function on the server side which satisfies the \textit{Geo-Ind} requirements after user customization i.e., removal of locations that do not satisfy the Boolean predicates.
In order to generate this \emph{robust} function efficiently, \systemName{} minimizes the number of constraints by using a \emph{graph approximation}. 
The experimental results on a real dataset show that the robust obfuscation function generated by \systemName{} is customizable with only minimal loss in utility compared to the traditional approaches which are not robust against customization.

The main contributions of this work are  as follows:

\begin{itemize}
    \item  We propose a tree-based approach to assist users in the specification of customization preferences. This improves the utility of location reporting as the number of locations in the obfuscation function is lower than traditional non-hierarchical  approaches \cite{qiu2020location}.
    \item We present a customization preferences model which are expressed in the form of Boolean Predicates and are selectively shared with server for the purpose of generating obfuscation function. Our customization model is more expressive than the prior work for customization based on policy graphs~\cite{cao2020pglp}. 
    \item We develop a novel method for generating obfuscation functions, that is \emph{robust} against user customization. 
    \item We implement graph approximation to reduce the number of constraints and thus make optimization problem for obfuscation function generation efficient.
    \item We design a framework with interactions between an untrusted server which performs computationally heavy tasks and a user device which performs tasks involving real location data.
    \item We evaluate \systemName{} on a real dataset (Gowalla - social network based on user check-ins) to show the effectiveness of the framework w.r.t privacy, utility, and customization.
\end{itemize}

The rest of the paper is organized as follows. We introduce the \systemName{} framework and describe the key concepts used in our work in Section~\ref{sect:background}. In Section~\ref{sect:models}, we present the tree-based representation used in this work along with the policy model. In Section~\ref{sect:generating}, we describe in detail the generation of the customized and robust obfuscation function for each user. We present in Section~\ref{sect:architecture}, architecture of our framework and detail the control flow on the user and server side. In Section~\ref{sect:experiments}, we evaluate our approach on a real dataset and compare it against a baseline. In Section~\ref{sect:related_work} we go over the related work and we conclude the work by summarizing our contributions, and possible future extensions in Section~\ref{sect:conclusions}.

\section{Preliminaries} 
\label{sect:background}
In this section, we introduce the  \systemName{} %
framework (Section \ref{subsec:framework}) and the preliminaries (Section \ref{subsec:preliminaries}) of our geo-obfuscation approach. 

\subsection{Background}
\label{subsec:preliminaries}
In this section, we formalize  key concepts and notions for our proposed framework, introduced above. %

\vspace{-0.00in}
\begin{table}[h]
\caption{Main notations and their descriptions}
\vspace{-0.10in}
\label{Tb:Notationmodel}
\centering
\small 
\begin{tabular}{l l}
\hline
\hline
Symbol                  & Description \\
\hline
$\vNode{}{i}$                   & Location $i$ or node $i$ \\
$\mathcal{V}$           & Set of locations or nodes  \\
$\vPrior{\vNode{}{i}}$      & Prior probability of location i \\
$\vDistance{i}{j}$   & Distance between locations/nodes $\vNode{}{i}$ and $\vNode{}{j}$ \\ 
$\vNodeSet{}^k$         & The set of nodes with height $k$ in the location tree \\
$\vNodeSet{}^0$         & The set of leaf nodes in the location tree \\
$\vTree{}^i$              & A location tree with $\vNode{}{i}$ as root node \\
$\children{\vNode{}{i}}$ & Children nodes of location $i$ \\ 
$\mathbf{Z}^{K}$      & Obfuscation Matrix at level K \\
$\vCell{}{i,j}$     & Entry in the matrix at row i and column j \\
$\prunesize$          & Number of location nodes to be pruned \\
$\vPruneSet{}$         & Set of location nodes to be pruned \\
\hline
\end{tabular}
\normalsize
\end{table}
\vspace{-0.00in}

\begin{figure*}
    \centering
    \includegraphics[width=0.60 \linewidth]{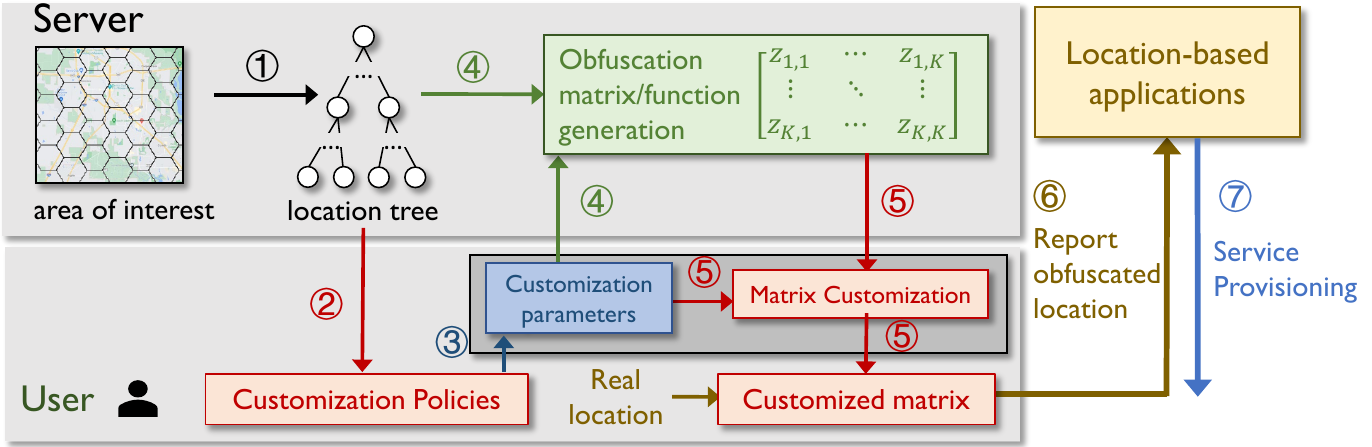}
    \caption{Overview of CORGI framework.}
    \label{fig:overview}
    \vspace{-0.1in}
\end{figure*}

\noindent \textbf{Obfuscation matrix}. Generally, when considering the obfuscation range as a finite discrete location set $\mathcal{V} = \left\{v_1, ..., v_K\right\}$, an obfuscation strategy can be represented as a stochastic matrix $\mathbf{Z} = \left\{z_{i,j}\right\}_{K \times K}$ \cite{qiu2020location}. Here, each $z_{i,j}$ represents the probability of selecting $v_j \in \mathcal{V}$ as the obfuscated location given the real location $v_i \in \mathcal{V}$. For  each real location $v_i$ (corresponding to each row $i$ of $\mathbf{Z}$), the \emph{probability unit measure} needs to be satisfied: 
\begin{equation}
\label{eq:probunitmeasure}
\textstyle    \sum_{j=1}^K z_{i, j} = 1, \forall i = 1, ..., K,
\end{equation}
i.e., the sum probability of its obfuscated locations is equal to 1. In this paper, we consider the location set $\mathcal{V}$ at different granularity levels, and any real location can be only obfuscated to the locations at the same granularity level (details are introduced in Section \ref{subsec:locationtree}).

\vspace{0.05in}
\noindent \textbf{Privacy Criteria}. 
From the attacker's perspective, the user's actual and reported locations can be described as two random variables $X$ and $Y$, respectively. We apply \emph{Geo-Indistinguishability (Geo-Ind)} \cite{andres2013geo} as the privacy criterion for location privacy guarantees:  
\begin{definition}
($\epsilon$-Geo-Ind) Given the obfuscation matrix $\mathbf{Z}$ that covers a set of locations $\mathcal{V}$ at the same granularity level, $\mathbf{Z}$ is called $\epsilon$-Geo-Ind if only if for each pair of real locations $v_i, v_j \in \mathcal{V}$ and any obfuscated $v_l \in \mathcal{V}$
    \begin{equation}
    \label{equ:Geo-Ind}
        \frac{\mathbb{P}\mathrm{r}\left(X= v_i \left|Y = v_l\right.\right)}{\mathbb{P}\mathrm{r}\left(X = v_j \left| Y= v_l\right.\right)} \leq e^{\epsilon d_{i,j}}\frac{p_{v_i}}{p_{v_j}}, 
    \end{equation}
where $p_{v_i}$ and $p_{v_j}$ denote the prior distributions of $v_i$ and $v_j$, respectively, $\epsilon > 0$ is predetermined constant called privacy budget, and $d_{i,j}$ denotes the distance between $v_i$ and $v_j$. %

\end{definition}

Equ. (\ref{equ:Geo-Ind}) indicates that  the posterior of the user's location estimated from its obfuscated location is close to the user's prior location distribution and how close they are depends on the parameter $\epsilon$. In  other words, an external  attacker cannot obtain sufficient additional information from a user's obfuscated location.  

The utility of our approach is measured based on the estimation error in travelling distance due to using obfuscated location in service provisioning.
Given that user's real location is $\vNode{}{i}$, the obfuscated location generated is $\vNode{}{l}$, the target location is $\vNode{}{n}$, the utility is given by
\begin{equation}
\textstyle U(\vNode{}{i}, \vNode{}{l}, \vNode{}{n}) = \mid \vDistance{\vNode{}{i}}{\vNode{}{n}} - \vDistance{\vNode{}{l}}{\vNode{}{n}} \mid. 
\end{equation}
where $\vDistance{\vNode{}{i}}{\vNode{}{n}}$ is implemented using haversine formula. If there are multiple target locations denoted by ${\vNode{}{1}, \ldots, \vNode{}{N}}$, the overall utility is computed as
$\frac{1}{N}\sum_{n=1}^{N} U(\vNode{}{i}, \vNode{}{j}, \vNode{}{n})$.

\subsection{Framework}
\label{subsec:framework}

Our problem setting is that of \emph{Location Based Services (LBS)} where users share their privatized locations with a server in order to receive service provisioning (e.g. Uber, Lyft, Yelp, Citizen Science). There are three main actors in our setting: \emph{users}, \emph{third party providers}, and a \emph{server}. 
\emph{Users} wish to share their locations in a privacy-preserving manner with applications. 
They specify policies in order to state their customization preferences and have a privacy module/middleware running on their mobile device or on a trusted edge computer to assist with location hiding. 
\emph{Third party providers} use the privatized locations shared by the user for providing services to the user.
Finally, we have the \emph{server} which runs on the cloud with whom non-sensitive portions of the user preferences are shared and it takes care of computationally heavy operations.
Users do not trust neither the third party providers nor the server with their sensitive location information or preferences.
Figure~\ref{fig:overview} introduces the flow of \systemName{} and interactions among these three actors: 
\newline \textcircled{1} The server generates a spatial index/location tree for an area of interest that contains the real location of the user (Section~\ref{subsec:locationtree}).  
\newline \textcircled{2}\textcircled{3} The location tree is shared with the users to allow them to specify their preferences (Section~\ref{subsec:policymodel}).
\newline \textcircled{4} The server obtains the customization parameters relevant for determining the privacy budget and generating a robust obfuscation function which guarantees Geo-Indistinguishability \cite{andres2013geo}, and thus provide strong location privacy guarantees. Obfuscated function is represented by a set of probability distributions in an \textit{obfuscation matrix} (Section~\ref{subsect:feasibility}). 
\newline \textcircled{5} Users receive the obfuscation function/matrix and customize it based on their needs (Section~\ref{subsect:pruning})\footnote{Matrix customization operations can be done on a trusted edge server if user device lacks the computational capability}. 
\newline \textcircled{6} \textcircled{7} This customized obfuscation function is utilized to determine the user's obfuscated location,  to be  shared with  third party location-based applications for the purpose of service provisioning.

\section{Models}
\label{sect:models}

In this section, we introduce the models, including the location tree model (Section \ref{subsec:locationtree}), i.e., how we organize locations at different granularity levels in a tree structure, and the user customization policies (Section \ref{subsec:policymodel}), i.e., what attributes are considered in the customization.
\vspace{-0.00in}
\subsection{Location Tree Model}
\label{subsec:locationtree}

We build a hierarchical index over a given spatial region for location representation. 
We  design a tree-like structure, called \emph{location tree}, where each level of the tree represents a particular granularity of location data, and lower levels of the tree increase  granularity. 
This representation of locations is intuitive and makes it easier for users to specify the granularity of location sharing they are comfortable with.

In general, a tree can be represented by $\mathcal{T} = \left(\vNodeSet{}, \ordering \right)$, where $\vNodeSet{}$ denotes the node set and $\ordering$ describes the ordered relationship between nodes, i.e., $\forall \vNode{}{i}, \vNode{}{j}  \in \vNodeSet{}$, $\vNode{}{j}\ordering \vNode{}{i}$ means that $\vNode{}{j}$ is a child of $\vNode{}{i}$. $\forall \vNode{}{i} \in \vNodeSet{}$, we let $\children{\vNode{}{i}}$ denote the set of $\vNode{}{i}$'s children, i.e., $\children{\vNode{}{i}} = \left\{\vNode{}{j} \in \vNodeSet{}\left|\vNode{}{j} \ordering \vNode{}{i}\right.\right\}$. Here, we slightly abuse notation by letting $v_i$ denote both location $i$ and its corresponding node in the location tree. Given these notations, we formally define a location tree  as follows: 
\begin{definition}
\label{def:locationtree}
(\textbf{Location Tree}) A location tree $\mathcal{T} = \left(\vNodeSet{}, \ordering\right)$ is a rooted tree, where  
\begin{itemize}
    \item the root node $\vNode{}{\mathrm{r}} \in \vNodeSet{}$ represents the whole area
    \item the tree is balanced and  leaf nodes are $\left\{\vNode{}{1}, \ldots, \vNode{}{K}\right\}$; 
    \item for each non-leaf node $\vNode{}{i} \in \vNodeSet{}$, its children $\vNode{}{j} \in \children{\vNode{}{i}}$ represent a \emph{partition} of $\vNode{}{i}$, i.e.,  locations in $\children{\vNode{}{i}}$ are disjoint and their union is $\vNode{}{i}$. 
\end{itemize}
\end{definition}

We partition the node set $\vNodeSet{}$ in the location tree into $H+1$ levels: $\vNodeSet{}^{0}$, $\ldots$, $\vNodeSet{}^{H}$, where $H$ is the height of the tree (i.e., the number of hops from the node to the deepest leaf). $\vNodeSet{}^{0}$ represents the set of \emph{leaf nodes}. 
We define \emph{obfuscation in a location Tree} $\left(\vNodeSet{}, \ordering\right)$ as a function that maps a given real location $\vNode{}{i} \in \vNodeSet{}^n$  to another location $\vNode{}{k} \in \vNodeSet{}^n$ and both nodes are at the same level $n$.

We generate this location tree using  Uber's H3\footnote{https://eng.uber.com/h3/} hexagonal hierarchical spatial index which takes as input the longitude, latitude, and resolution (between 0 and 15, with 0 being coarsest and 15 being finest), and outputs a hexagonal grid index for the region (as illustrated in Figure~\ref{fig:treegeneration}). 
H3 divides the target region into contiguous hexagonal cells of the same size using the given resolution level.
For each of the cells generated by H3, it keeps the distance between the center point of a hexagon and its neighboring cells consistent, making it a better candidate for represeting spatial relationships than grid based systems such as Geohash~\footnote{http://geohash.org/site/tips.html}.
Figure~\ref{fig:treegeneration} illustrates the location tree generated for Times Square, New York. Blue nodes at highest granularity represent the leaf nodes. Red and green nodes at lower granularity represent the intermediate nodes. The root node encompasses the entire region. Nodes in each level do not overlap with each other.
Our approach for location representation is inspired by previous works on spatial indexing such as R-Tree proposed by Beckmann et. al \cite{beckmann1990r}.
In  Beckmann's  approach, however,   location nodes can overlap and are not disjoint partitions.
In our cases, if any two location nodes overlap, it is hard to  assess whether these two locations satisfy $\epsilon$\textit{-Geo-Ind} (Equation~(\ref{equ:Geo-Ind})) or not as a user can be in both of these nodes simultaneously. 

\begin{figure}
    \centering
    \includegraphics[width=1.00 \linewidth]{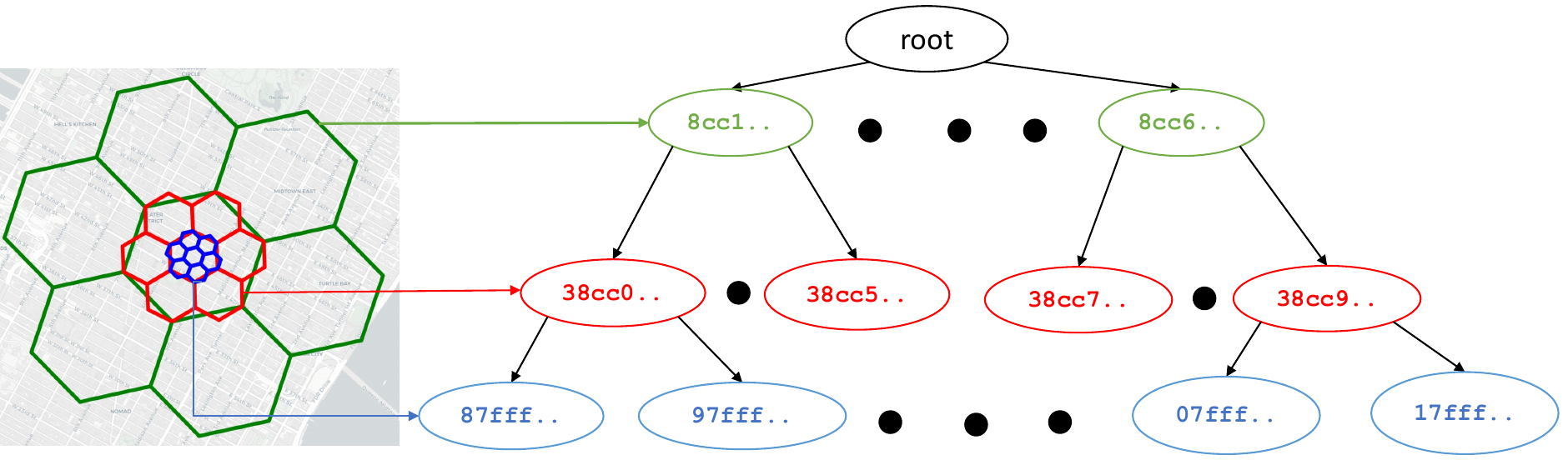}
    \caption{Location Tree  of 3 levels generated using H3. }
    \label{fig:treegeneration}
\end{figure}

\subsection{User Customization Policies}
\label{subsec:policymodel}

Users express customization requirements by way of policies. These policies help determine the properties of the final obfuscation function that is generated. 
A policy captures users' customization requirements as follows: \newline \centerline{$<Privacy\_l$, $\; Precision\_l$, \; $User\_Preferences>$}

\vspace{0.05in}
\noindent \textbf{Privacy level} or {\em $Privacy\_l$} is a user-set parameter that determines the obfuscation range i.e., set of locations/nodes from which users' obfuscated location is selected. 
Given a policy $\vPolicySet{}$ with privacy\_l = $n$, a \emph{privacy forest} is the set of all sub-trees with nodes at level $n$ as their root. 
Thus, the privacy forest contains all the possible locations that can be reported as obfuscated location.
As Figure~\ref{fig:location_tree} shows, if a user selects privacy level $n$,  we first determine the nodes at height $n$ $\vNodeSet{}^n$ which forms the privacy forest.
Accordingly, a higher privacy level implies a wider range of obfuscated locations to select for users.
In Figure~\ref{fig:location_tree}, the red and blue colored subtrees indicate two different user policies both of which specify their privacy level as 2 but with different real locations. 
For a particulart user, if $\vNode{}{}$ is the ancestor of user's real location at height $n$, the sub-tree with $\vNode{}{}$ as the root node includes all the locations that the user could report.
The server can use \emph{privacy level} to limit the number of locations in the obfuscation matrix for this user, and accordingly, reduce the overhead of generating it as well as improve the utility of location reporting compared to traditional approaches~\cite{qiu2020location}. The \textit{privacy level} also provides the flexibility for the user to specify the range of locations they are comfortable sharing. \\
\noindent \textbf{Precision level} ({\em Precision\_l}) specifies the exact granularity at which the user reports their locations (e.g., neighborhood or block). 
For example, if a user requires the precision level to be $1$, then his reported location/node is restricted to the set of nodes in level i.e., $\vNodeSet{}^1$.
Thus Precison\_l gives users the flexibility to reduce the granularity at which location is shared depending on their needs.
As the privacy level is the maximum possible granularity for location sharing, precision level is always lower than the privacy level.\\
\noindent \textbf{User Preferences} specify  
users' preferred options for location selection and further narrows down the obfuscation range and therefore reduces the number of locations/nodes in the matrix.
These may be  expressed in a variety of ways, depending on the application at hand and the users' requests (e.g.black lists of locations, dynamic checks etc). An intuitive approach is to encode preferences as Boolean predicates in the form $<$ \textit{var, \; op, \; val}$>$ where \textit{var} denotes commonly used preferences for locations such as home, office, traffic, weather, driving\_distance etc; \textit{op} is one among $\{=, \neq, <, >, \geq, \leq \}$ depending upon the variable; and \textit{val} is assigned from the domain of the \textit{var}. \\
An example of a policy modelled using these 3 attributes is as follows: 
$<$\textit{privacy\_l = 3}, \textit{precision\_l = 0}, \textit{user\_preferences = [popular = ``True'', distance $\leq$ 5 miles}]
This customization policy states that user would prefer to have the \emph{privacy forest} with nodes from level 3 (privacy\_l = 3) and the nodes in this privacy forest represents their obfuscation range. 
From these set of possible locations, any of them which are not popular i.e., locations where people usually gather,
and has a distance higher than 5 miles from their real location should not be considered for reporting (user\_preferences).
Finally, when generating their obfuscated location, they would like it to be at the granularity of level 0 (precision\_l = 0) which are the leaf nodes.

\begin{figure}
    \centering
    \includegraphics[width=1.00 \linewidth]{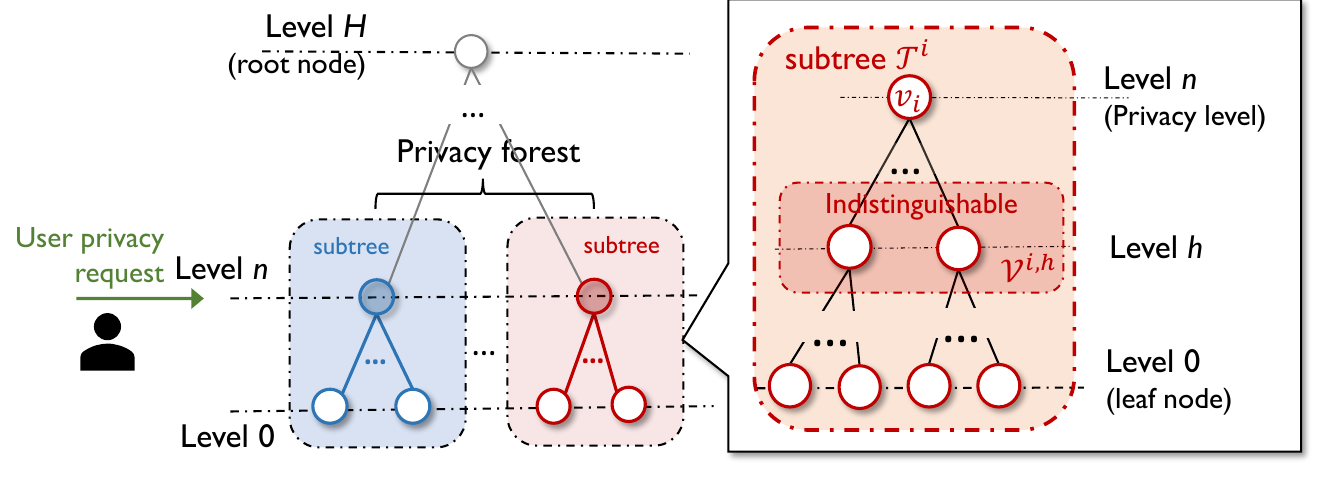}
    \vspace{-0.10in}
    \caption{Tree-based geo-obfuscation.}
    \label{fig:location_tree}
    \vspace{-0.10in}
\end{figure}

\vspace{-0.0in}
\section{Generating Robust Obfuscation Matrix}
\label{sect:generating}

We describe how to generate a robust obfuscation matrix, that preserve strong privacy guarantees while meeting users' customization policies, using the location tree.  
This is non-trivial, as introducing additional constraints based on policies affects the ability to obfuscate locations within certain regions and limit the range of possible obfuscated locations.
\vspace{-0.00in}
\subsection{Feasibility Conditions for Geo-I}
\label{subsect:feasibility}
\vspace{-0.00in}
Given a user's requested privacy level $n$, the nodes at level $n$ ($\vNode{}{i} \in \vNodeSet{}^n$) and their children nodes ($\children{\vNode{}{i}}$) represent the possible set of obfuscated locations for a user.
As the server does not know the user's real location or the subtree that contains user's real location, it has to generate obfuscation matrix for each node $\vNode{}{i}$ based on the leaf nodes in $\children{\vNode{}{i}}$.
The server then returns all the generated obfuscation matrices to the user, and the user selects the obfuscation matrix according to their real location. 
Suppose users want to report a location with lower granularity than the leaf nodes for which the matrix is generated, they can do by applying precision reduction (discussed in Section \ref{sect:precisionreduction}) to generate the obfuscation matrix at the desired precision level.

\DEL{
\begin{proposition}
\label{prop:convexity}
(Convexity) The feasible region of obfuscation matrices is convex. 
\end{proposition}
\begin{proof}
The detailed proof can be found in Appendix. 
\end{proof}}

\begin{figure}
    \centering
    \includegraphics[width=0.95\linewidth]{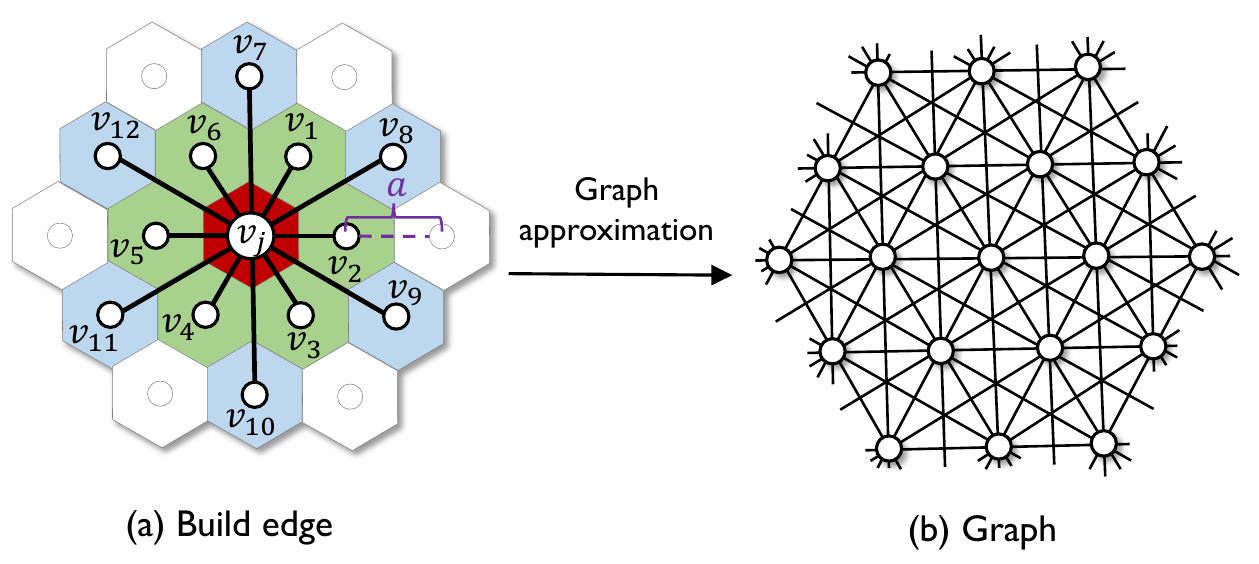}
    \vspace{-0.00in}
    \caption{Graph approximation. \\
    {\small *(a) Build 12 edges connected to $v_1$: $e_{j, 1}, ..., e_{j, 12}$. The distance from $v_1, ..., v_6$ to $v_j$ is $a$; and the distance from $v_7, ..., v_{12}$ to $v_j$ is $\sqrt{3}a$.}}
    \label{fig:graphapproximation}
    \vspace{-0.0in}
\end{figure}

Next, we introduce how to generate {\em feasible}  obfuscation matrices  of each subtree rooted at level $n$. Suppose that $v_i$ is a node at level $n$, then we use $\mathcal{T}^{i}$ to denote the subtree rooted at $v_i$ and let $\mathcal{V}^{i,0}$ denote the set of leaf nodes in $\mathcal{T}^{i}$, as shown in Fig. \ref{fig:location_tree}.

We use $\mathbf{Z}^{0} = \left\{z_{k,l}\right\}_{\left|\mathcal{V}^{i,0}\right|\times \left|\mathcal{V}^{i,0}\right|}$ to represent the obfuscation matrix of $\mathcal{T}^{i}$ at precision level $0$ (the highest precision level). We call $\mathbf{Z}^{0}$ is \emph{feasible} if only if both $\epsilon$-Geo-Ind (general case is defined in Equ. (\ref{equ:Geo-Ind})) 
\begin{equation}
\label{equ:Geo-Ind1}
    \frac{\mathbb{P}\mathrm{r}\left(X= v_j \left|Y = v_l\right.\right)}{\mathbb{P}\mathrm{r}\left(X = v_k \left| Y= v_l\right.\right)} \leq e^{\epsilon d_{i,j}}\frac{p_{v_j}}{p_{v_k}}, \forall v_j, v_k, v_l \in \mathcal{V}^{i,0}
\end{equation}
and probability unit measure 
\begin{equation}
\label{equ:probabilityunit1}
\textstyle    \sum_{v_l \in \mathcal{V}^{i,0}} z_{k, l} = 1,~\forall k \in \mathcal{V}^{i,0},
\end{equation}
are satisfied. We let $\mathcal{Q} = {v_1, ..., v_M}$ denote a set of places of interests which in our problem setting are locations where service is requested for for e.g., passenger pickup. 
Given the target location $v_q \in \mathcal{Q}$, the actual location $v_k$ of a user, the obfuscated location $v_l$, the expected estimation error of moving distance caused by obfuscation matrix $\mathbf{Z}^{0}$ is obtained by 
\vspace{-0.0in}
\begin{equation}
\label{eq:matrix_gen_1}
\textstyle \Delta_q\left(\mathbf{Z}^{0}\right) = \sum_{v_k\in \mathcal{V}^{i,0}} \mathbb{P}\mathrm{r}\left(X = v_{k}\right) \sum_{v_l\in \mathcal{V}^{i,0}} z_{k,l} \textstyle U(\vNode{}{k},\vNode{}{l},\vNode{}{q}).
\end{equation}

Given the probability distribution of target locations $\mathbb{P}\mathrm{r}\left(Q=v_q\right)$, then we can define the quality loss as the expected estimation error of moving distance as
\begin{equation}
\label{eq:matrix_gen}
\textstyle \Delta\left(\mathbf{Z}^{0}\right) = \sum_{v_q \in \mathcal{V}^{i,0}} \mathbb{P}\mathrm{r}\left(Q=v_q\right)\Delta_n\left(\mathbf{Z}^{0}\right) %
\end{equation}
$\mathbf{Z}^{0}$ is generated by solving the following linear programming (LP) problem
\begin{eqnarray}
\label{eq:LP}
    \min ~ \Delta\left(\mathbf{Z}^{0}\right) &\mbox{s.t.} ~ \mbox{Equ. (\ref{equ:Geo-Ind1}) (\ref{equ:probabilityunit1}) are satisfied}
\end{eqnarray}%
i.e., minimize the expected estimation error of moving distance using the matrix $\mathbf{Z}^0$ to all the target locations. 
Once $\mathbf{Z}^{0}$ is generated, it will be delivered to the user and they are allowed to customize $\mathbf{Z}^{0}$ based on evaluation of \textit{User\_Preferences} (Section \ref{subsect:pruning}) and selection of the desired granularity level (Section \ref{sect:precisionreduction}). 

While customizing, users select to remove certain number of locations from the obfuscation range and this results in a pruning of the matrix.
Note that, after $\mathbf{Z}^{0}$ is pruned, the new matrix might no longer satisfy the Geo-I constraints in Equ. (\ref{equ:Geo-Ind1}) (the details of matrix pruning will be introduced in Section \ref{subsect:pruning}). Intuitively, to avoid this potential privacy issue,  we need to reserve more privacy budget when formulating the $\epsilon$-Geo-Ind constraints, and generate a more \emph{robust} matrix that allows users to remove up to a certain number of locations in the matrix without violating Geo-I. The details of generating such robust matrix will be given in Section \ref{subsec:robustmatrix}.    

\subsection{Graph approximation to reduce number of Geo-I constraints}
\label{subsec:graphapprox}
According to the definition of Geo-Ind in Equ. (\ref{equ:Geo-Ind1}), for each column (location) of the obfuscation matrix $\mathbf{Z}^{0}$, an $\epsilon$-Geo-Ind constraint is generated for pairwise comparison of all locations, leading to a total of $O(K^3)$ constraints. This generates a very high computation load to derive $\mathbf{Z}^0$. 
To improve the time efficiency of the matrix calculation, we approximate the users' mobility on the 2D plane by a graph, where it is sufficient to enforce $\epsilon$-Geo-Ind for each pair of neighboring nodes (\emph{Theorem \ref{thm:transitivity}}), to enforce the $\epsilon$-Geo-Ind constraints for all pairs of nodes. This reduces the number of constraints in LP from $O(K^3)$ to $O(12 \times K^2) = O(K^2)$.

The method for approximating the hexagonal grid to a graph is illustrated in Figure~\ref{fig:graphapproximation}. We connect each node $v_i$ to not only the 6 immediate neighbors (denoted by $v_1, ..., v_6$) but also the 6 other diagonal neighbors (denoted by $v_7, ..., v_{12}$). We let $a$ denote the distance between the immediate neighbors, computed based on the distance between their center points and therefore the weight of each edge is set to $a$. Then, we can obtain a weighted graph $\mathcal{G}$ as Fig. \ref{fig:graphapproximation}(b) shows. 
The length of the shortest path between any pair of nodes $v_j$ and $v_k$ on the graph, denoted by $d_{\mathcal{G}}(v_j, v_k)$.
Since the graph is undirected, we have $d_{\mathcal{G}}(v_j, v_k) = d_{\mathcal{G}}(v_k, v_j)$, $\forall v_j, v_k \in \mathcal{V}^{i,0}$. 
To ensure that $\epsilon$-Geo-Ind on $\mathcal{G}$ to be a sufficient condition of the original $\epsilon$-Geo-Ind constraint defined on the 2D plane, we need to guarantee that  $d_{\mathcal{G}}(v_j, v_k)$ is no longer than their Euclidean distance $\vDistance{j}{k}$, i.e., $d_{\mathcal{G}}(v_j, v_k) \leq d_{j,k}$ (the reason will be further explained in the proof of the \emph{Theorem \ref{thm:transitivity}}). We first introduce Lemma \ref{lem:graph} as a preparation of of Theorem \ref{thm:transitivity}. 

\begin{lemma}
\label{lem:graph}
$\forall v_j, v_k \in \mathcal{V}^{i, 0}$, $d_{\mathcal{G}}(v_j, v_k) \leq d_{j,k}$. 
\end{lemma}
\begin{proof}
We consider the two locations $v_k$ and $v_j$ on a polar coordinate system, where $v_j$ is located at the origin point. We use $[r_k, \varphi_k]$ to represent $v_k$'s polar coordinate, where $r_k \geq 0$ denotes the radial coordinate and $\varphi_k \in (-\pi, \pi]$ denotes the angular coordinate. As Fig. \ref{fig:graphapproximation_proof} shows, there are 6 different cases according to the value of $\varphi_k$: Case 1: $\varphi_k\in \left(-\frac{\pi}{6}, \frac{\pi}{6}\right]$, Case 2: $\varphi_k\in \left(\frac{\pi}{6}, \frac{\pi}{2}\right]$, Case 3:  $\varphi_k\in \left(\frac{\pi}{2}, \frac{5\pi}{6}\right]$, Case 4: $\varphi_k\in \left(\frac{5\pi}{6}, -\frac{5\pi}{6}\right]$, Case 5: $\varphi_k\in \left(-\frac{5\pi}{6}, -\frac{\pi}{2}\right]$, and Case 6: $\varphi_k\in \left(-\frac{\pi}{2}, -\frac{\pi}{6}\right]$. In what follows, we prove that Lemma \ref{lem:graph} is true in Case 1, where the conclusion can be applied to other 5 cases due to the symmetricity of the 6 cases.  

Case 1 can be further divided to the two cases: Case 1(a), when $\varphi_k \in \left[0, \frac{\pi}{6}\right]$, and Case 1(b), when $\varphi_k \in \left[\frac{11\pi}{6}, 0\right]$. 
\newline In Case 1(a), we can always find a location $v_1$ that is $u$ hops away from $v_j$ in the direction of $\frac{\pi}{6}$ and $w$ hops away from $v_k$ in the direction of $\frac{2\pi}{3}$. Starting from $v_1$ and $v_k$, if we move in the direction of $\frac{2\pi}{3}$, we can find a location $v_2$ and $v_3$ with the radial coordinate equal to 0. 
The number of hops from $v_1$ to $v_j$ is equal to number of hops from $v_2$ to $v_j$ ($u$ hops). The number of hops from $v_k$ to $v_1$ is equal to number of hops from $v_3$ to $v_2$ ($w$ hops). Note that the length of each hop in the graph is $a$. In Case 1(b), we can always find a location $v_1$ that is $u$ hops away from $v_j$ in the direction of $-\frac{\pi}{6}$ and $w$ hops away from $v_k$ in the direction of $-\frac{\pi}{3}$. Similarly, we can find the corresponding $v_2$ that is $u$ hops away from $v_j$ and $v_3$ that is $w$ hops away from $v_2$. 
In both Case 1(a)(b), according to the \emph{Law of Sines}, we obtain that \begin{equation}
    d_{j,k} = \frac{\sin \angle v_j v_3 v_k}{\sin \angle v_j v_k v_3} d_{j,3} \geq d_{j,3} = \left(u+w\right)a
\end{equation}
from which we can then derive that (according to the \emph{triangle inequality} on a graph)
\begin{eqnarray}
    d_{\mathcal{G}}\left(v_j, v_k\right) &\leq&   \underbrace{d_{\mathcal{G}}\left(v_j, v_1\right)}_{u\times a} + \underbrace{d_{\mathcal{G}}\left(v_1, v_k\right)}_{w\times a} %
    \leq d_{j, k} \nonumber
\end{eqnarray}
The proof is completed. 
\end{proof}

\begin{figure}[t]
    \centering
    \includegraphics[width=0.90 \linewidth]{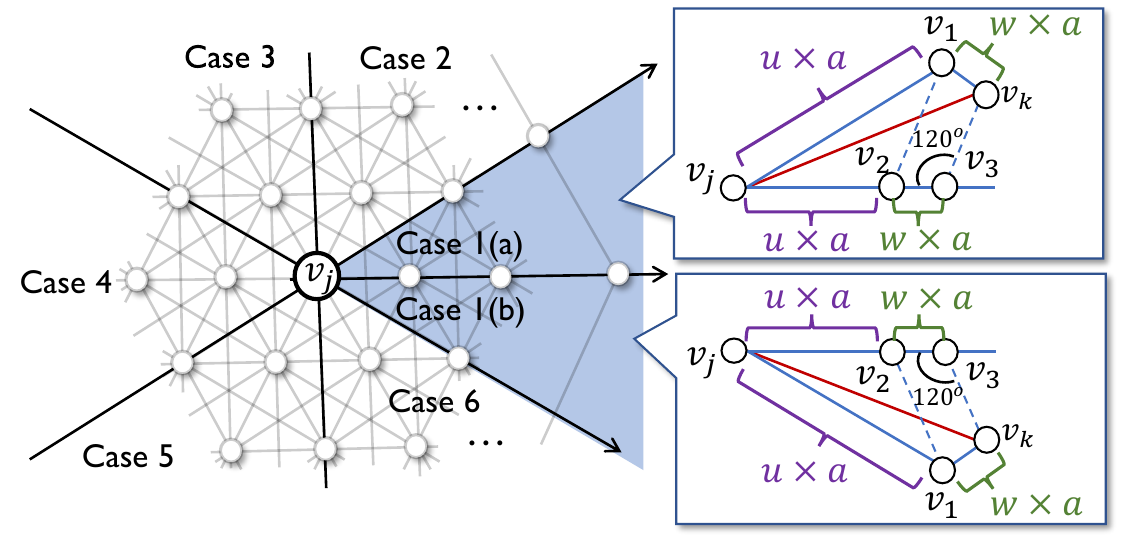}
    \caption{Proof of Lemma \ref{lem:graph}}
    \label{fig:graphapproximation_proof}
    \vspace{-0.10in}
\end{figure}

\begin{theorem}
\label{thm:transitivity}
(\emph{Transitivity of $\epsilon$-Geo-Ind}) To enforce $\epsilon$-Geo-Ind for each pair of locations, it is sufficient to enforce $\epsilon$-Geo-Ind only for each pair of neighboring peers in the graph $\mathcal{G}$. 
\end{theorem}
\begin{proof}
We pick up any pair of locations. Without loss of generality, we denote the two locations by $(v_1, v_M)$ and denote their shortest path by $\mathcal{S}_{(v_1, v_M)} = \left(\left(v_1, v_2\right),  ... , \left(v_{M-1}, v_M\right)\right)$. We then prove that $(v_1, v_M)$ satisfies $\epsilon$-Geo-Ind if all the neighboring peers satisfy Geo-I. 

Since $v_{1}, ..., v_{M}$ are in the shortest path from $v_1$ to $v_M$ sequentially, $\vDistance{1}{M} \geq d_{\mathcal{G}}(1, M) = \sum_{l=1}^{M-1} d_{\mathcal{G}}(v_l, v_{l+1})$ (according to Lemma \ref{lem:graph}).  
\DEL{Then, as $(v_1, v_{m_1})$ and $(v_{m_1}, v_M)$ satisfy $\epsilon$-Geo-Ind, 
\begin{eqnarray}
\label{eq:geoI1m1}
\Rightarrow && z_{1,k}  - e^{{\epsilon d_{1,m_1}}}  z_{m_1,k} \leq 0, ~\forall k\\
\label{eq:geoIm1M}
&& z_{m_1,k}  - e^{{\epsilon d_{m_1,M}}}  z_{M,k} \leq 0, ~\forall k
\end{eqnarray}
from which we can derive that $\forall k$,
\begin{eqnarray}
&& z_{1,k}  - e^{{\epsilon d_{1,M}}}  z_{M,k}  \\ \nonumber 
&=& \underbrace{\left(z_{1,k}  -  e^{{\epsilon d_{1,m_1}}}  z_{m_1,k}\right)}_{\leq 0~\mbox{\footnotesize according to Equ. (\ref{eq:geoI1m1})}} +  \underbrace{\left(e^{{\epsilon d_{1,m_1}}}  z_{m_1,k}  - e^{{\epsilon \left(d_{1, m_1} + d_{1, M}\right)}}z_{M,k}\right)}_{\leq 0~\mbox{\footnotesize according to Equ. (\ref{eq:geoIm1M})}} \\
&\leq& 0. 
\end{eqnarray}
indicating $(v_1, v_{M})$ satisfies Geo-I.} 

Because each neighboring peer  $(v_{m_{l}}, v_{m_{l+1}})$ ($l = 1, ..., M-1$) satisfies $\epsilon$-Geo-Ind, for each obfuscated location $v_k$, 
\begin{eqnarray}
&& z_{1,k}  - e^{{\epsilon d_{1,M}}}  z_{M,k} \leq  z_{1,k}  - e^{{\epsilon \sum_{l=1}^{M-1} d_{\mathcal{G}}(v_l, v_{l+1})}}  z_{M,k}\\ \nonumber 
&=&\sum_{l = 1}^{M-1}\underbrace{\left(z_{l,k} - e^{{\epsilon d_{l, l+1}}}z_{l+1,k}\right)}_{\leq 0~\mbox{\footnotesize since $(v_{l}, v_{l+1})$ satisfy $\epsilon$-Geo-Ind}}e^{{\epsilon \sum_{h=1}^{l-1}d_{h, h+1}}} \leq 0, 
\end{eqnarray}
indicating that $(v_1, v_M)$ satisfy $\epsilon$-Geo-Ind. The proof is completed. 
\end{proof}

Note that enforcing $\epsilon$-Geo-Ind for neighbors in $\mathcal{G}$ provides a sufficient condition for the original $\epsilon$-Geo-Ind constraints (defined in Equ. (\ref{equ:Geo-Ind})), but not a necessary condition, which means it might shrink the feasible region of the original LP defined in Equ. (\ref{eq:LP}), leading to a higher quality loss ($\Delta\left(\mathbf{Z}^{0}\right)$).

\subsection{Customization by Matrix Pruning}
\label{subsect:pruning}

\begin{figure}
    \centering
    \includegraphics[width=0.96 \linewidth]{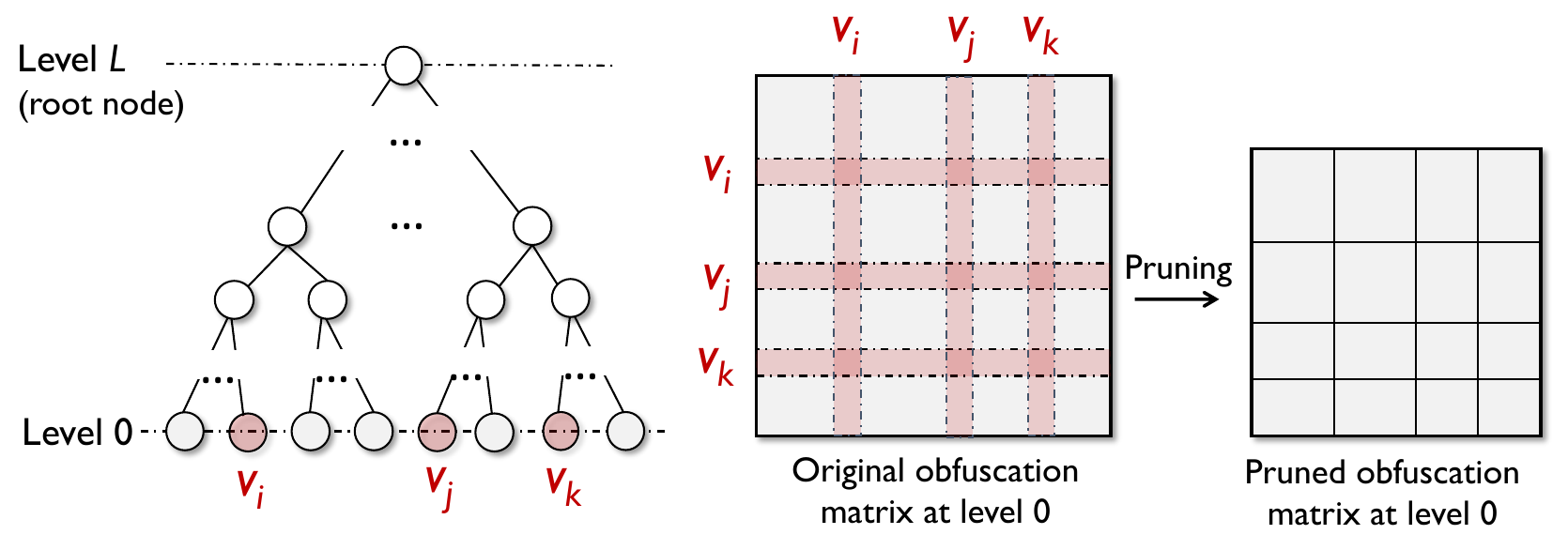}
    \vspace{-0.0in}
    \caption{Matrix pruning (in the figure, $\vNode{}{i}, \vNode{}{j}, \vNode{}{k}$ are the locations to be pruned).}
    \label{fig:pruning}
    \vspace{-0.0in}
\end{figure}

After receiving obfuscation matrices from the server, the user can select the matrix $\mathbf{Z}^{0}$ based on their real location and can customize the matrix by removing the locations that do not satisfy their preferences. simplicity. 
For example, in Figure~\ref{fig:pruning}, the three nodes marked in red at Level 0, \{$\vNode{}{i}, \vNode{}{j}, \vNode{}{k}  $\}, are to be pruned\footnote{Even though pruning can be done at any level of the tree, it makes the most sense to do it for locations at leaf node (at highest granularity) so as to remove only the exact locations and avoid over-pruning.}. The 3 corresponding rows and columns in the matrix $\mathbf{Z}^{0}$ are highlighted and in the next step, they are removed. 

The resulting  matrix $\vMatrix{0}{*}$ is considered  \textit{feasible}, only if it still satisfies the probability unit measure for each row in the matrix as per Equ.~(\ref{eq:probunitmeasure}). ~%
We denote the set of nodes %
(that do not satisfy the user's preferences)
to be removed from the matrix by $\vPruneSet{}$ ($\vPruneSet{} \subseteq \vNodeSet{}^{0}$).
After pruning, the new obfuscation matrix $\vMatrix{0}{*}$ is of dimensions $m \times m$ where $m = \vCount{\vNodeSet{}^{0} - \vPruneSet{}}$. This process called \emph{matrix pruning} is carried out as follows:

\begin{itemize}
\item Remove the rows and the columns of nodes with indices in $\mathcal{S}$ from $\mathbf{Z}^{0}$ to create $\mathbf{Z}^{0}_{*}$.  
\item For each remaining row $i$ in $\vMatrix{0}{*}$, multiply each entry in the matrix $\vCell{}{i,k}$ by $\frac{1}{1 - \sum_{l \in \vPruneSet{}}\vCell{}{i,l}}$, i.e., $\vCell{}{i,k} \leftarrow \frac{\vCell{}{i,k}}{1 - \sum_{l \in \vPruneSet{}}\vCell{}{i,l}}$. 
\end{itemize}

This ensures that the entries in each row still satisfy the probability unit measure, i.e.,  
\begin{eqnarray}
    \sum_{k\in \vNodeSet{}^0\backslash \vPruneSet{}} \vCell{}{i,k} = \frac{\sum_{k\in \vNodeSet{}^0\backslash \vPruneSet{}} \vCell{}{i,k}}{1 - \sum_{l \in \vPruneSet{}}\vCell{}{i,l}} = \frac{\sum_{k\in \vNodeSet{}^0} \vCell{}{i,k}-\sum_{k\in \vPruneSet{}} \vCell{}{i,k}}{1 - \sum_{l \in \vPruneSet{}}\vCell{}{i,l}} = 1. \nonumber 
    \label{eq:prunedUnitProb}
\vspace{-0.0in}
\end{eqnarray}

\subsection{Ensuring Robustness of Customized Matrix}
\label{subsec:robustmatrix}

After  matrix pruning, although the pruned matrix satisfies the probability unit measure, it might not satisfy $\epsilon$-Geo-I since in each column $k$, the entries $z_{i,k}$ ($i = 1, ..., K$) are multiplied by different factors $\frac{1}{1 - \sum_{l \in \vPruneSet{}}\vCell{}{i,l}}$.
We denote the size of the set of nodes to be pruned from the matrix as $\delta$ i.e., $\delta$=$\vCount{\vPruneSet{}}$, and define \emph{$\delta$-prunable} robust matrix as follows:

\begin{definition}
\label{def:deltaPrune}
An obfuscation matrix $\mathbf{Z}$ is called \emph{$\delta$-prunable} if, after removing up to $\delta$ number of nodes from $\mathbf{Z}$ through \emph{matrix pruning}, the new matrix $\vMatrix{}{*}$ still satisfies $\epsilon$-Geo-Ind, i.e., $\forall i, j, k$, 
\begin{equation}
\label{equ:epsilonij}
    \frac{z_{i,k}}{1 - \sum_{l \in \mathcal{S}}z_{i,l}} - e^{\epsilon d_{i,j}}  \frac{z_{j,k}}{1 - \sum_{l \in \mathcal{S}}z_{j,l}} \leq 0, \forall \mathcal{S} \subseteq \mathcal{V}^{i,0} \mbox{ s.t. }\left|\mathcal{S}\right| \leq \delta
\vspace{-0.0in}
\end{equation}
\end{definition}

In order to make an obfuscation matrix $\delta$-prunable, we need to reserve more privacy budget $\epsilon_{i,j}$ (defined in Equ. (\ref{eq:robustGeoI})) for each pair of locations $v_i$ and $v_j$, such that even a certain number of locations are pruned from the matrix, the Geo-I constraints of $v_i$ and $v_j$ are still satisfied.
We now define \emph{reserved privacy budget}, denoted by $\epsilon_{i,j}$, as follows.

\begin{definition}
The reserved privacy budget $\epsilon_{i,j}$ for each pair of locations $\vNode{}{i}$ and $\vNode{}{j}$ where $i, j$ are their indices in the obfuscation matrix is given by, 
\vspace{-0.00in}
\begin{equation}
\label{eq:robustGeoI}
    \epsilon_{i,j} = \frac{1}{d_{i,j}}\ln \left(\max_{\mathcal{S} \subseteq \mathcal{V}^{i,0} \mbox{ s.t. } |\mathcal{S}| \leq \delta} \frac{1 - \sum_{l \in \mathcal{S}}z_{j,l}}{1 - \sum_{l \in \mathcal{S}}z_{i,l}}\right)
\vspace{-0.00in}
\end{equation}
\end{definition}

\begin{proposition}
\label{prop:matrixprune}
A sufficient condition for $\mathbf{Z}$ to be $\delta$-prunable is to satisfy 
\vspace{-0.10in}
\begin{equation}
\label{eq:prunableconstr}
   z_{i,k} - e^{\left({\epsilon } - \epsilon_{i,j}\right)d_{i,j}}z_{j,k}\leq 0, \forall i, j, k.
\vspace{-0.10in}
\end{equation}
\end{proposition}
\begin{proof}
Given that Equation (\ref{eq:prunableconstr}) is satisfied, then for each column $k$, $\forall i, j, \mathcal{V}'_0 \in \mathcal{V}_0$ with $|\mathcal{V}'_0| \leq \delta$, 
\normalsize

\vspace{-0.00in}
\begin{eqnarray}
&&\frac{z_{i,k}}{1 - \sum_{l \in \mathcal{V}'_0}z_{i,l}} - e^{{\epsilon}d_{i,j}}  \frac{z_{j,k}}{1 - \sum_{l \in \mathcal{V}'_0}z_{j,l}} \nonumber \\
&=& \frac{1}{1 - \sum_{l \in \mathcal{V}'_0}z_{j,l}}\left(\frac{1 - \sum_{l \in \mathcal{V}'_0}z_{j,l}}{1 - \sum_{l \in \mathcal{V}'_0}z_{i,l}}z_{i,k} - e^{{\epsilon }d_{i,j}}z_{j,k}\right) \nonumber \\
&\leq& \frac{1}{1 - \sum_{l \in \mathcal{V}'_0}z_{j,l}}\left(e^{\epsilon_{i,j}d_{i,j}}  z_{i,k} - e^{{\epsilon }d_{i,j}}z_{j,k}\right) \nonumber \\
&=& \frac{e^{\epsilon_{i,j}d_{i,j}}}{1 - \sum_{l \in \mathcal{V}'_0}z_{j,l}}\underbrace{\left(z_{i,k} - e^{\left({\epsilon } - \epsilon_{i,j}\right)d_{i,j}}z_{j,k}\right)}_{\leq 0~\mbox{\footnotesize according to Equ. (\ref{eq:prunableconstr})}}\leq 0. \nonumber
\end{eqnarray}
\end{proof}
\vspace{-0.00in}

Thus, we can state the minimization problem for \emph{robust matrix generation}, $\min ~\Delta\left(\mathbf{Z}^0\right)$, where the objective function (Equ.~(\ref{eq:matrix_gen})) and equality constraints remains the same as earlier (Equ.~(\ref{equ:probabilityunit1})) but the inequality constraints are updated to Equ.~(\ref{eq:prunableconstr}) using reserved privacy budget.

In order to calculate $\epsilon_{i,j}$ in Equ.~(\ref{eq:robustGeoI}), we need to consider all the possible subsets of $\mathcal{S} \subseteq \mathcal{V}^{i,0}$ with the cardinality no larger than $\delta$.
The complexity of computing the reserved privacy budget increases exponentially with $\delta$. 
Therefore, we define an approximation of $\epsilon_{i, j}$, denoted by $\epsilon'_{i,j}$ as follows:
\begin{equation}
\label{eq:approximateepsilon}
    \epsilon'_{i,j} = \frac{1}{d_{i,j}}\ln \left( \frac{1 - \frac{\max_{\mathcal{S} \subseteq \mathcal{V}^{i,0} \mbox{ s.t. } |\mathcal{S}| \leq \delta}\sum_{l \in \mathcal{S}}z_{j,l}}{e^{\epsilon d_{i,j}}}}{1 - \max_{\mathcal{S} \subseteq \mathcal{V}^{i,0} \mbox{ s.t. } |\mathcal{S}| \leq \delta}\sum_{l \in \mathcal{S}}z_{j,l}}\right)
\end{equation}

\begin{proposition}
\label{prop:rpb_approximation}
The matrix generated by replacing $\epsilon_{i,j}$ with $\epsilon'_{i,j}$ in Equ.~(\ref{eq:prunableconstr}) is an upper bound of the solution.
\end{proposition}
\begin{proof}
\begin{eqnarray}
&& \epsilon_{i,j} \nonumber \\
& = & \frac{1}{d_{i,j}}\ln \left(\max_{\mathcal{S} \subseteq \mathcal{V}^{i,0} \mbox{ s.t. } |\mathcal{S}| \leq \delta} \frac{1 - \sum_{l \in \mathcal{S}}z_{j,l}}{1 - \sum_{l \in \mathcal{S}}z_{i,l}}\right) \nonumber\\ 
&\leq& \frac{1}{d_{i,j}}\ln \left(\max_{\mathcal{S} \subseteq \mathcal{V}^{i,0} \mbox{ s.t. } |\mathcal{S}| \leq \delta} \frac{1 - \frac{\sum_{l \in \mathcal{S}}z_{i,l}}{e^{\epsilon d_{i,j}}}}{1 - \sum_{l \in \mathcal{S}}z_{i,l}}\right) ~(\mbox{as $e^{\epsilon d_{i,j}} z_{j,l} \leq z_{i,l} $}) \nonumber\\
&\leq &  \frac{1}{d_{i,j}}\ln \left(\frac{1 - \frac{\max_{\mathcal{S} \subseteq \mathcal{V}^{i,0} \mbox{ s.t. } |\mathcal{S}| \leq \delta}\sum_{l \in \mathcal{S}}z_{i,l}}{e^{\epsilon d_{i,j}}}}{1 - \max_{\mathcal{S} \subseteq \mathcal{V}^{i,0} \mbox{ s.t. } |\mathcal{S}| \leq \delta}\sum_{l \in \mathcal{S}}z_{i,l}}\right) = \epsilon'_{i,j} \nonumber
\end{eqnarray}
\normalsize
\end{proof}

To calculate $\epsilon'_{i,j}$, we need to find the top $\delta$ number of $z_{j,l}$ with $v_l \in \mathcal{V}^{i,0}$, which takes $O(K \log K)$ in the worst case. 
According to Proposition \ref{prop:rpb_approximation}, by replacing  $\epsilon_{i,j}$ with $\epsilon'_{i,j}$ in Equ. (\ref{eq:prunableconstr}), we can obtain a sufficient condition of Equ. (\ref{eq:prunableconstr}). 

\begin{equation}
\label{eq:aec}
   z_{i,k} - e^{\left({\epsilon } - \epsilon'_{i,j}\right)d_{i,j}}z_{j,k}\leq 0, \forall i, j, k. 
\end{equation}

By replacing Equ. (\ref{eq:prunableconstr}) with this sufficient condition expressed in Equ. (\ref{eq:approximateepsilon}), we have the robust matrix generation problem which is an upper bound on the solution.

\begin{eqnarray}
\label{eq:RLP}
    \min ~ \Delta\left(\mathbf{Z}^{0}\right) &\mbox{s.t.} ~ \mbox{Equ. (\ref{eq:aec}) (\ref{equ:probabilityunit1}) are satisfied}
\end{eqnarray}

\begin{algorithm}[h]
\small 
    \SetAlgoLined
    \DontPrintSemicolon
    \SetKwFunction{FMain}{generateRobustMatrix}
    \SetKwProg{Fn}{Function}{:}{}  
    \Fn{\FMain{$\vNodeSet{}$, $Prob^0$, $\delta$, $\epsilon$, $t$}}{
        $i$ = $0$\; 
        $\vMatrix{}{i}[0, 0] \ldots [\vCount{ \vNodeSet{} } - 1, \vCount{\vNodeSet{}} - 1]$ = 0 \\ 
        $\vMatrix{}{i}$ = LPSolver($\vNodeSet{}$, $\epsilon$, $Prob^0$) \; \Comment{Matrix generated by solving Equ. (\ref{eq:LP})} \\
        $RPB[0, 0] \ldots [\vCount{ \vNodeSet{} } - 1, \vCount{\vNodeSet{}} - 1]$ = 0 \\
        \Do{$i \leq t$}{
            $i$ +=1 \; 
            $RPB$ = computeRPB($\vNodeSet{}{}$, $\vMatrix{}{i}$, $\delta$) \; \Comment{Reserved Privacy Budget (RPB) using Equ. (\ref{eq:approximateepsilon})} \\
            $\vMatrix{}{i}$ = LPSolver($\vNodeSet{}$, $\epsilon$, $Prob^0$, $RPB$) \; 
            \Comment{Matrix generated by solving Equ. (\ref{eq:RLP})} \\
}
\Return $\vMatrix{}{t}$\; 
}
    \caption{Robust matrix generation}
    \label{alg:robustmatrix}
\end{algorithm}

Algorithm~\ref{alg:robustmatrix} takes as input the set of nodes $\vNodeSet{}$, their prior probability distribution $Prob^0$, number of locations to be pruned $\delta$, privacy parameter $\epsilon$, and the number of iterations for convergence $t$ (which is determined empirically based on convergence experiments, see Section~\ref{sect:experiments}).
The non-robust matrix is generated first by solving the linear programming problem expressed in Equ. (\ref{eq:LP}) (Step 4).
For storing the Reserved Privacy Budget (RPB) for each pair of locations $\vNode{}{i}$ and $\vNode{}{j}$, we initialize a matrix denoted by RPB (Step 6).
We iteratively compute the RPB matrix using Equ. (\ref{eq:approximateepsilon}) and then use it to generate the matrix using the linear programming problem expressed in Equ. (\ref{eq:RLP}). This process is repeated for $t$ iterations until the RPB matrix, as well as the matrix generated using it, converges.
The robust obfuscation matrix is returned in the final step.

\subsection{Matrix Precision Reduction}
\label{sect:precisionreduction}

\begin{figure}
    \centering
    \includegraphics[width=1.00 \linewidth]{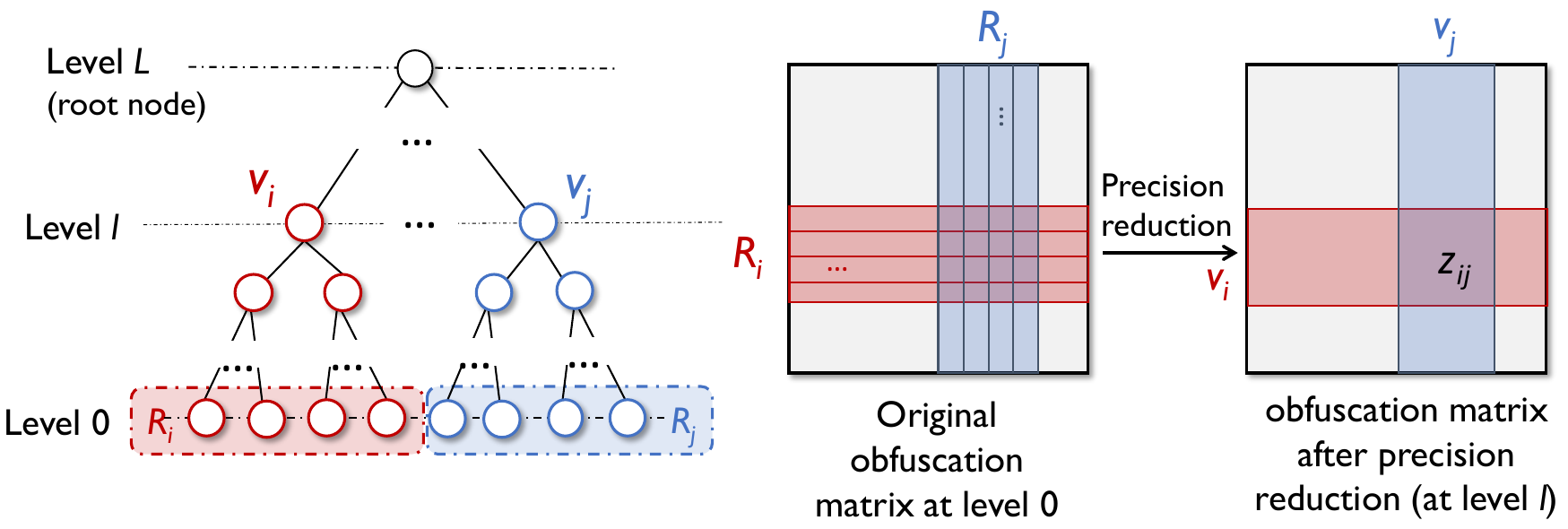}
    \vspace{-0.10in}
    \caption{Matrix precision reduction.}
    \label{fig:precisionreduction}
    \vspace{-0.10in}
\end{figure}

In Figure~\ref{fig:precisionreduction}, the original obfuscation matrix is generated for level 0, i.e., the set of leaf nodes. 
Suppose the user specifies a value $l$ as its  Precision\_l. \emph{Matrix precision reduction} generates the obfuscation matrix at level $l$, $\vMatrix{l}{}$ ($l > 0$), given the obfuscation matrix at level $0$, $\vMatrix{0}{}$.
As illustrated in the figure, the new matrix is generated by replacing all the rows of the descendant leaf nodes with their corresponding ancestor nodes at level $l$.
For each pair of nodes $\vNode{}{i}$ and $\vNode{}{j}$ at level $l$, we use $\children{\vNode{}{i}}$ and $\children{\vNode{}{j}}$ to represent the set of their descendant leaf nodes, respectively. 
The probability of selecting $\vNode{}{j}$ as the obfuscated location given the real location $\vNode{}{i}$ is calculated using Bayes' theorem.
\vspace{-0.00in}

\begin{equation}{\label{eq:precisionReduction}}
\vCell{l}{i,j} 
= \frac{\sum_{\vNode{}{m} \in \children{\vNode{}{i}}} \vPrior{\vNode{}{m}} \sum_{\vNode{}{n} \in \children{\vNode{}{j}}} \vCell{0}{m, n}}{\vPrior{\vNode{}{i}}}
\end{equation}
where $\vPrior{\vNode{}{m}}$ and $\vPrior{\vNode{}{i}}$ denote the prior distributions of $\vNode{}{m}$ and $\vNode{}{i}$ respectively. Note that, $\vPrior{\vNode{}{i}} = \sum_{\vNode{}{m} \in \children{\vNode{}{i}} } \vPrior{\vNode{}{m}}$.

\begin{proposition}
\label{prop:precisionreduction}
Matrix precision reduction preserves both probability unit measure and $\epsilon$-Geo-Ind. 
\end{proposition}

\begin{proof}
First, we check the probability unit measure. 
\begin{eqnarray}
\mathcal{V}_0 \nonumber = \cup_{s_j \in \mathcal{V}_l} \mathcal{R}_j 
    \Rightarrow \sum_{s_j \in \mathcal{V}_l} \sum_{s_v \in \mathcal{R}_j}z^{0}_{u,  v} = \sum_{s_v \in \mathcal{V}_0}z^{0}_{u,  v} = 1. 
\end{eqnarray}
We take sum of the entries in each row $i$ in $\mathbf{Z}^l$, 
\begin{eqnarray}
\sum_{s_j \in \mathcal{V}_l}z^{l}_{i,j} &=&  \sum_{s_j \in \mathcal{V}_l} \frac{\sum_{s_u \in \mathcal{R}_i} p_{u}  \sum_{s_v \in \mathcal{R}_j}z^{0}_{u,  v}}{p_{i}} \nonumber\\  
&=& \frac{\sum_{s_u \in \mathcal{R}_i} p_{u}  \left(\sum_{s_j \in \mathcal{V}_l}\sum_{s_v \in \mathcal{R}_{j}}z^{0}_{u,  v}\right)}{p_{i}} \nonumber \\
&=& \frac{\sum_{s_u \in \mathcal{R}_i} p_{u}}{p_{i}} = \frac{p_i}{p_{i}} = 1, \nonumber
\end{eqnarray}
i.e., each row $i$ satisfies the probability unit measure. 

\DEL{
\begin{eqnarray}
&& \sum_{s_v \in \mathcal{R}_j}\sum_{s_u \in \mathcal{R}_i} p_{u}p_v \left(e^{{\epsilon}d_{i,j}} - e^{{\epsilon}d_{u,v}}\right) \\
&=& e^{{\epsilon}d_{i,j}}\sum_{s_v \in \mathcal{R}_j}\sum_{s_u \in \mathcal{R}_i} p_{u}p_v \left( 1- e^{{\epsilon}\left(d_{u,v} - d_{i,j}\right)}\right) \\
&\leq & e^{{\epsilon}d_{i,j}}\sum_{s_v \in \mathcal{R}_j}\sum_{s_u \in \mathcal{R}_i} p_{u}p_v {\epsilon}\left(d_{i,j} - d_{u,v}\right) 
\end{eqnarray}}

We then check $\epsilon$-Geo-Ind for each column $k$  in $\mathbf{Z}^l$: $\forall s_i, s_j$
\begin{eqnarray}
&& z^{l}_{i,k} - e^{{\epsilon}}  z^{l}_{j,k} \nonumber \\ 
& = & \frac{\sum_{s_u \in \mathcal{R}_i} \sum_{s_w \in \mathcal{R}_k}p_{u} z^{0}_{u, w}}{p_{i}} - e^{{\epsilon}}   \frac{\sum_{s_v \in \mathcal{R}_j} \sum_{s_w \in \mathcal{R}_k}p_{v} z^{0}_{v,  w}}{p_{j}} \nonumber \\ 
&=& \sum_{s_w \in \mathcal{R}_k} \left(\frac{\sum_{s_v \in \mathcal{R}_j}\sum_{s_u \in \mathcal{R}_i} p_{u}p_v \left(z^{0}_{u,  w}- e^{{\epsilon}} z^{0}_{v, w}\right)   }{p_{i}p_{j}}\right) \nonumber \leq 0 \nonumber 
\end{eqnarray}
since $z^{0}_{u,  w}- e^{{\epsilon}} z^{0}_{v, w} \leq 0$ $\forall u, v, w$. 
\end{proof}

\begin{algorithm}[t]
    \SetAlgoLined
    \DontPrintSemicolon
    \SetKwFunction{FMain}{precisionReduction}
    \SetKwProg{Fn}{Function}{:}{}  
    \small 
    \KwInput{Obfuscation matrix (at level 0) $\vMatrix{0}{}$, Location Tree $\vTree{}$,  Precision Level $l$}
    \KwOutput{Obfuscation Matrix (at level l) $\vMatrix{l}{}$}
    \Fn{\FMain{$\vMatrix{0}{}$, $\vTree{}$, $l$}}{
        $\vNodeSet{}^l$ = getNodes($\vTree{}$, $l$) \Comment{Get nodes at precision level} \\ 
        $\vMatrix{l}{}[0, 0] \ldots [\vCount{ \vNodeSet{}^l } - 1, \vCount{\vNodeSet{}^l} - 1]$ = 0 \\
        \For{$i$ $\in$ $0, \ldots, \vCount{\vNodeSet{}^l} -1$}
        {   
            \For{$j$ $\in$ $0, \ldots, \vCount{\vNodeSet{}^l} - 1$}
            {
                    $num = 0, den$ = 0 \\
                    \For{$u$ $\in$ $0 \ldots \mid \children{\vNode{}{i}} \mid - 1$}
                    {
                    $row\_sum$ = 0 \\
                    \For{$v$ $\in$ $0 \ldots \mid \children{\vNode{}{j}} \mid - 1$}
                    {
                    $row\_sum$ = $row\_sum + \vCell{0}{u, v}$
                    } 
            $num$ = $num + \vPrior{\vNodeSet{}^0}[u] \times row\_sum$ \\
            $den = den + \vPrior{\vNodeSet{}^0}[u]$  \\
            }
            $\vCell{l}{i, j}$ = $\frac{num}{den}$ \\
            }
        }
        \Return{$\vMatrix{l}{}$}
    }
    \caption{Precision Reduction Function}
    \label{alg:precisionReduction}
\end{algorithm}

Algorithm~\ref{alg:precisionReduction} presents the approach for matrix precision reduction given the matrix for leaf nodes ($\vMatrix{0}{}$), the location tree ($\vTree{}$), the prior probability distribution of the leaf nodes $prob^0$, and the precision level ($l$) which specifies the granularity/height of the tree of the reported location.
First, we get the set of nodes ($\vNodeSet{}^l$) from level $l$ (Step 2).
We initialize the new obfuscation matrix with dimensions based on the set of nodes retrieved.
For each pair of location nodes ($\vNode{}{i}, \vNode{}{j}$) in $\vNodeSet{}^l$, we compute their corresponding probability in the new matrix ($\vCell{l}{i,j}$) by using the probabilities of their leaf nodes, $\children{\vNode{}{i}}$ and $\children{\vNode{}{j}}$ respectively, in Equ.~(\ref{eq:precisionReduction}) (Steps 4-16).
The prior probability distribution for the leaf nodes $\vPrior{\vNodeSet{}^0}$ in this subtree can be obtained by querying the server\footnote{We assume that the prior probability distribution is readily available based on publicly available information. We explain how it is computed for a real dataset in Section~\ref{sect:experiments}. We disregard communication and computation cost for this as it is a relatively small vector and only has to be periodically updated.}.
Finally, the new matrix $\vMatrix{l}{}$ for level $l$ is returned.
Thus using matrix precision reduction, we are able to save the overhead of generating the obfuscation matrix when the user chooses to share at a lower granularity than the leaf nodes.

\section{Location Obfuscation by CORGI}
\label{sect:architecture}
In this section, we describe the steps performed \systemName{} in order to generate the obfuscated location of a user.
This process is  sketched in Figure~\ref{fig:control_flow} and we explain the steps on the user and server side in detail below.

\begin{figure}
    \centering
    \includegraphics[width=1.00 \linewidth]{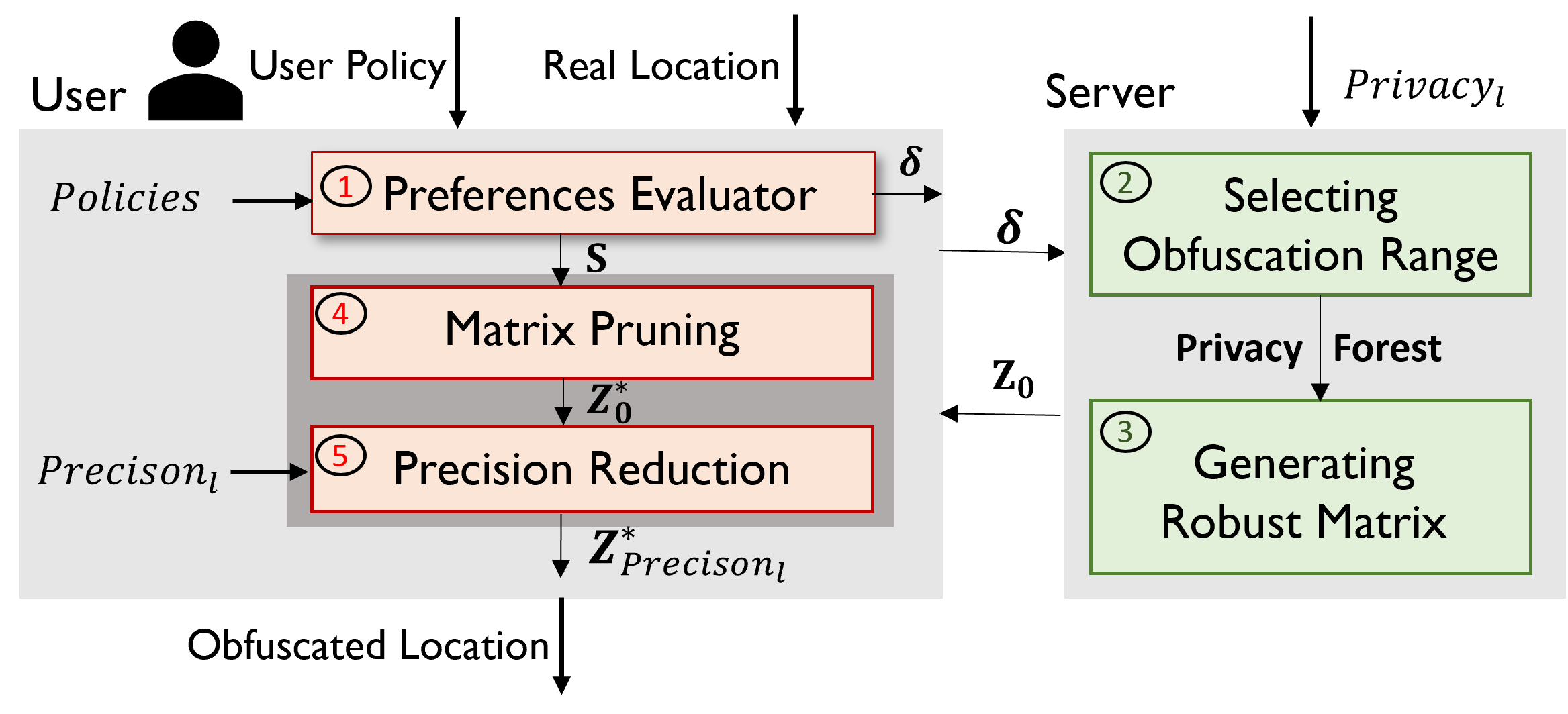}
    \vspace{-0.0in}
    \caption{Steps in generating the obfuscated location}
    \label{fig:control_flow}
    \vspace{-0.0in}
\end{figure}

\subsection{Server side}

\vspace{-0.0in}
\begin{algorithm}[h]
    \SetAlgoLined
    \DontPrintSemicolon
    \SetKwFunction{FMain}{generateMatrix}
    \SetKwProg{Fn}{Function}{:}{} 
    \small 
    \KwInput{Location Tree $\vTree{}$, Privacy Level $l$, Prune Parameter $\delta$}
    \KwOutput{Privacy Forest $PF$}
    \Fn{\FMain{$\vTree{}$, $l$, $\delta$}}{
        $\vNodeSet{}^l$ = getNodes($\vTree{}$, $l$) \\
        \textit{PF} = \{ \} \\
        \For{$\vNode{}{i}$ $\in \vNodeSet{}^l$}
        {
            $\vTree{}^i$ = findSubTree($\vNode{}{i}$, $\vTree{}$, $l$) \\
            $\vNodeSet{}^0$ = getNodes($\vTree{}^i$, $0$) \\
            $Prob^0$ = getPriorProbDist($\vNodeSet{}^0$) \\
            $\vMatrix{0}{}$ = generateRobustMatrix($\vNodeSet{}^0$, 
            $Prob^0$, $\delta$, $\epsilon$, $t$) \\
            \textit{PF[$\vTree{}^i$]} = $\vMatrix{0}{}$ \\ 
        }
        \Return{\textit{PF}}
    }
    \caption{Generate Robust Obfuscation Matrices based on Obsfucation Range}
    \label{alg:genObfuscation}
\end{algorithm}
\vspace{-0.0in}

\systemName{} on server side takes as input the Privacy level \textit{Privacy\_l }and the number of locations to be pruned ($\delta$).
Algorithm~\ref{alg:genObfuscation} describes the steps for determining the obfuscation range (represented by the privacy forest) and generating the obfuscation matrix.
First, the server determines the nodes at the given privacy level $\vNodeSet{}^{l}$ by performing a Breadth First Search in the Location Tree $\vTree{}$ (Step 2).
Second, it initializes the privacy forest as a dictionary where the key is a subtree and the value is the obfuscation matrix for the leaf nodes of that subtree (Step 3). 
The system then iterates through each node $\vNode{}{i}$ at the privacy level and generates an obfuscation matrix for each of them.
For this purpose, the server has to determine the subtree rooted at $\vNode{}{i}$ and perform a Depth First Search to determine the leaf nodes of that subtree (Steps 5-6). 
Next, it calls the \emph{generateRobustMatrix} (Algorithm~\ref{alg:robustmatrix}) with the set of leaf nodes and the number of locations to be pruned (Step 7) \footnote{The privacy parameter ($\epsilon$) and the number of iterations to convergence when generating the robust matrix ($t$) are set universally for all the users. $\epsilon$ is updated based on the user's customization needs.}.
The robust obfuscation matrix  $\vMatrix{0}{}$ thus generated for the leaf nodes is added to the dictionary with the subtree as the key (Steps 8-9). After iterating through all the nodes at the set privacy level and generating obfuscation matrices for each of their descendant leaf nodes, the final privacy forest \textit{PF} is returned.

\subsection{User side}
Users input their policies as well as their real location in order to generate the obfuscated location.
First, \systemName{} on user side determine the subtree $\vTree{}^i$ that contains user's real location and is rooted at $\vPolicySet{}$.Privacy\_l (Step 1).
We slightly abuse the notation here as $\vNode{}{i}$ denotes the real location as well as the node in the tree that contains the actual user's location. 
Next, the user preferences are evaluated on the leaf nodes of that subtree and determine the set of nodes ($\vPruneSet{}$) that are to be pruned. (Step 2 of Fig \ref{fig:control_flow}).
The number of locations in this set ($\vCount{\vPruneSet{}}$) along with  $\vPolicySet{}$.Privacy\_l is passed to the server side (Step 4).
From the privacy forest returned by the server (based on Algorithm~\ref{alg:genObfuscation}), the obfuscation matrix $\vMatrix{0}{}$  corresponding to the subtree that contains the user's real location is selected (Step 5).
Next, the system prunes this matrix by calling the Matrix Pruning Algorithm (\textit{pruneMatrix}) with the set of nodes to be pruned (Step 6).
The pruned matrix $\vMatrix{0}{*}$ is updated to reflect the required granularity specified in $\vPolicySet{}$.Precision\_l (Step 7).
From this final matrix, the row corresponding to the node at $\vPolicySet{}$.Precision\_l, which contains the ancestor of the real location of the user is selected.
The obfuscated location $\vNode{}{j}$ is selected from the row by sampling based on the probability distribution (Steps 8).

\vspace{-0.1in}
\begin{algorithm}[h]
    \SetAlgoLined
    \DontPrintSemicolon
    \SetKwFunction{FMain}{generateObsfucatedLocation}
    \SetKwProg{Fn}{Function}{:}{}  
    \small 
    \KwInput{Location Tree $\vTree{}$, Real Location $\vNode{}{i}$, Policy $\vPolicySet{}$}
    \KwOutput{Obfuscated Location $\vNode{}{j}$}
    \Fn{\FMain{$\vNode{}{i}$, $\vPolicySet{}$}}{
        $\vTree{}^i$ = findSubTree($\vNode{}{i}$, $\vTree{}$, $\vPolicySet{}$.Privacy\_l) \\
        $\vPruneSet{}$ = eval($\vTree{}^i$, $\vPolicySet{}$.User\_Preferences) \\
        \textit{PF} = generateMatrix($\vTree{}$, $\vPolicySet{}$.Privacy\_l, $\vCount{\vPruneSet{}})$ \\
        $\vMatrix{0}{}$ = \textit{PF}[$\vTree{}^i$] \\
        $\vMatrix{*}{0}$ = pruneMatrix( $\vMatrix{0}{}$, $\vPruneSet{}$) \\
        $\vMatrix{*}{l}$ = precisionReduction($\vMatrix{*}{0}$, $\vTree{}^i$, $\vPolicySet{}$.Precision\_l) \\
        $\vNode{}{j}$ = \textit{sample(}$\vMatrix{*}{l}$[\textit{ancestor(}$\vNode{}{i}, \vPolicySet{}$.Precision\_l)])) \\
        \Return{$\vNode{}{j}$} \\ 
    }
    \caption{Generate Obfuscated Location}
    \label{alg:genObsLocation}
\end{algorithm}
\vspace{-0.0in}

\subsection{Discussion}
\label{subsec:discussion}

It is possible that when evaluating user preferences at the time of location sharing more than $\delta$ locations need to be pruned based on the user preferences.  
In such a situation, there are two options for customization: (1) Satisfy all the user preferences which results in a set of locations to be pruned $\vPruneSet{}$ where $\vCount{\vPruneSet{}} > $ $\delta$ which leads to Geo-Ind violation, (2) Satisfy some of the policies in $\vPolicySet{}$ such that $\vCount{\vPruneSet{}} \leq \delta$ locations which leads to \emph{policy violations} (there exists a location $\vNode{}{} \in \vMatrix{}{}$ such that it does not satisfy a policy in $\vPolicySet{}$). 
Both these violations may occur if based on the policies a large number of locations have to be pruned from the matrix i.e., $\vCount{\vPruneSet{}}$ is large.
In such a case, \systemName{} finds it impossible to meet the $\delta$ requested by the user as well as generate an obfuscation matrix that is robust.

In this work, we have used approximation techniques in order to reduce the number of constraints (see Section~\ref{subsec:graphapprox}, Section~\ref{subsec:robustmatrix}).
An alternative method is to incorporate optimization decomposition in the linear programming model itself (similar to ~\cite{qiu2020time}) which would lead to improvement in utility.
Currently, \systemName{} supports point queries and does not handle trajectory data. This can be extended by replacing the privacy notion of Geo-Indistinguishability with Trajectory Indistinguishability~\cite{trafficadaptorsigspatial22} and allowing users to customize locations along their trajectories while ensuring various semantic constraints are met (e.g., road networks). 
Local Differential Privacy (LDP) has recently emerged as an approach to avoid using a centralized server and perturb users' data locally before it leaves their device~\cite{kasiviswanathan2011can}. 
However, to the best of our knowledge most previous works that utilize LDP have mainly focused on releasing population statistics and not location privacy as presented in this paper~\cite{DBLP:journals/corr/abs-2008-03686}.

\section{Experiments}
\label{sect:experiments}

\subsection{Experimental setup}

\textbf{Datasets:} We use the Gowalla dataset~\cite{cho2011friendship} for our experiments. Gowalla  is a location-based social networking website where users share their locations by checking-in. 
The dataset includes 
 check-in information, which has the following attributes: \textit{[user, check-in time, latitude, longitude, location id]}. 
We sampled the user check-ins from the San Francisco (USA) region in the Gowalla dataset. We choose this region because it had a dense distribution of check-ins distributed over a large area.
Overall, this sample includes 38,523 check ins. 
We generated the root node which covers the entire region at resolution 6 followed by the children for this root node at resolution 7.
We repeated the process two more times and generated a tree of height 3 with 343 leaf nodes.
For generating customization policies, we analyzed the sample and came up with simple heuristics to identify a user's home, office, and their outlier locations (where the user visited rarely and at odd times). 
We also analyzed the number of check-ins per location in order to identify what locations are popular and at what times.
Using this metadata we generated realistic user preferences such as \textit{home = ``False'', outlier = ``False'', popular = ``True''} .

\noindent \textbf{Priors:} We computed prior probability for the leaf nodes in the generated location tree by counting number of user check-ins within that node.
For intermediate nodes (higher up in the tree), the prior was computed by aggregating the priors of its children nodes.

\noindent \textbf{Baseline:} We used the commonly used mechanism of \emph{linear programming (LP)} for implementing the baseline \cite{Bordenabe-CCS2014, Qiu-TMC2020, Wang-WWW2017}. We call this baseline as \textit{non-robust} because this mechanism is not robust against removal of locations from the obfuscation range on the user side.

\noindent \textbf{Implementation:} All of the algorithms were implemented in Mathlab.
We used the state-of-the-art Linear Programming tool kit from Mathlab.
The data and location tree were stored in main memory. The full implementation including scripts to run the experiments is available on Github\footnote{\url{https://github.com/User-Privacy/CORGI}}. 
The experiments were run on a 4-core machine with 256 GB ram.

\vspace{-0.00in}
\subsection{Experimental results}

\subsubsection{Convergence}
In this experiment, we test the convergence of the \emph{quality loss} of \systemName{}, measured as the mean estimation error of traveling distance (implemented using \emph{Haversine distance}) to all the target locations. We set $NR\_TARGET$ = 49 (number of target locations that are randomly selected from a list of leaf nodes), $\epsilon$ = 15 km$^{-1}$ and used the priors generated from the Gowalla dataset. We ran two sets of experiments: when $\delta = 2$ and $\delta = 4$, where $\delta$ is the number of locations that the user wishes to remove after customization.
In each group, we ran the experiment for 10 times, and depicted the convergence of the quality loss in Fig. \ref{fig:experiment1}(a)(c) (when $\delta =2$) and Fig. \ref{fig:experiment1}(b)(d) (when $\delta =4$). In all four figures, the $x$-axis denotes the iteration index. In Fig. \ref{fig:experiment1}(a)(b), the $y$-axis 
represents the quality loss, while in Fig. \ref{fig:experiment1}(c)(d), the $y$-axis represents the difference between quality loss in consecutive iterations.  Here, a lower value on the $y$-axis denotes better convergence as there is little difference between entries in the matrix after each round.
As illustrated in Figure~\ref{fig:experiment1}(a)(b)(c)(d), the differences between quality loss in consecutive iterations converges in approximately 4 iterations for both values of $\delta$. For the rest of the experiments, we terminate the program after 10 iterations.

\begin{figure}[t]
\centering
\hspace{-0.0in}
\begin{minipage}{0.22\textwidth}
\centering
  \subfigure[\small $\delta$ = 2 (objective value)]{
\includegraphics[width=1.00\textwidth, height = 0.110\textheight]{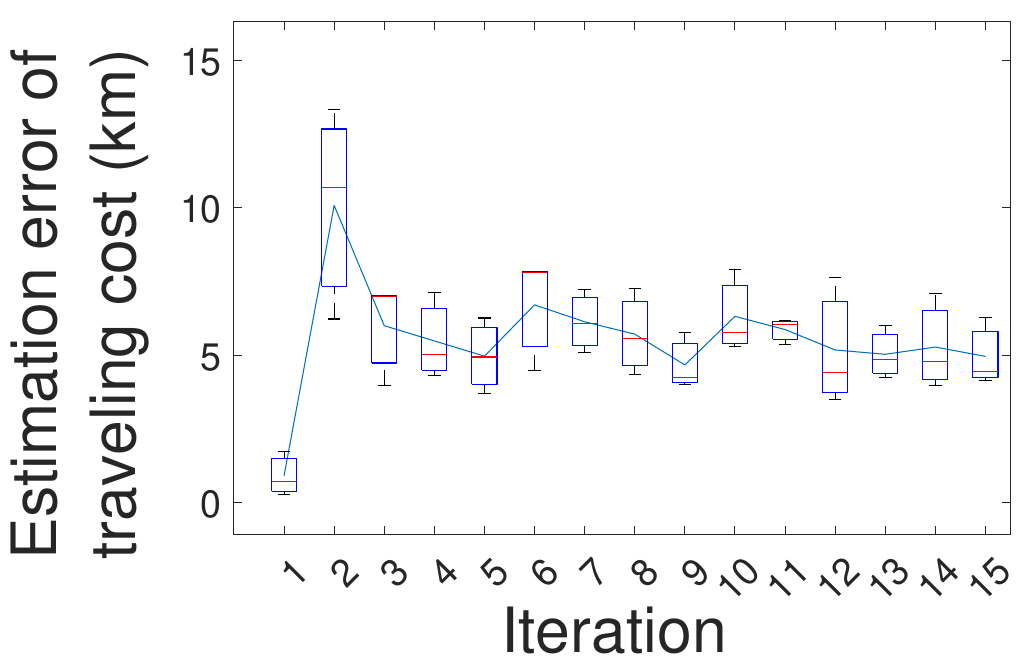}}
\label{fig:convergence2}
\end{minipage}
\hspace{0.1in}
\begin{minipage}{0.22\textwidth}
\centering
  \subfigure[\small $\delta$ = 4 (objective value)]{
\includegraphics[width=1.00\textwidth, height = 0.110\textheight]{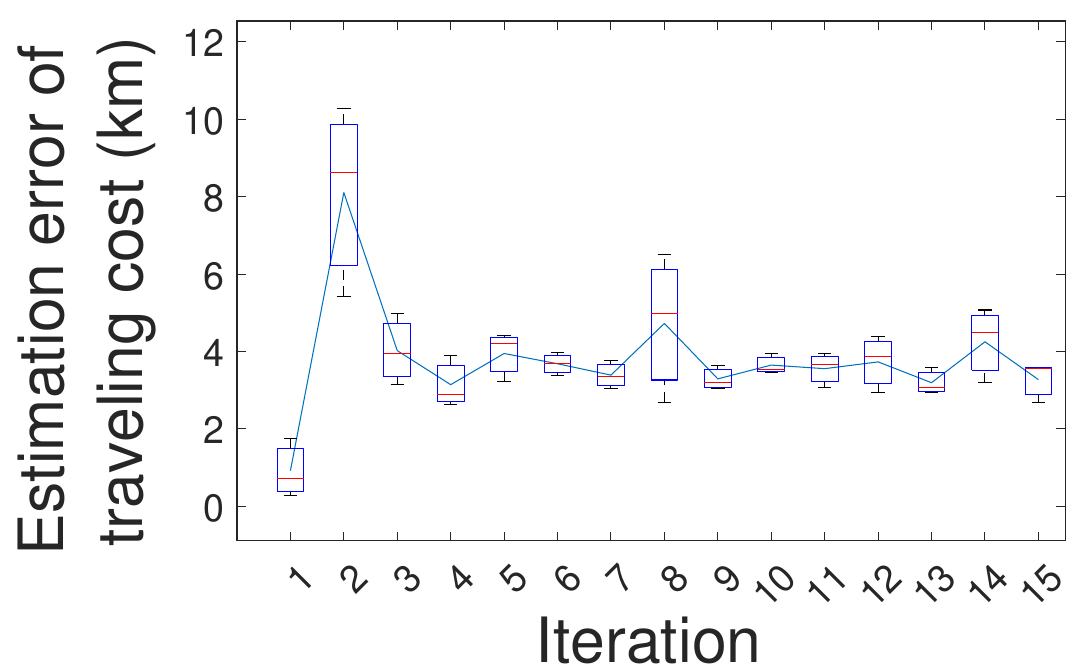}}
\label{fig:convergence1}
\end{minipage}
\begin{minipage}{0.22\textwidth}
\centering
  \subfigure[\small $\delta$ = 2 (difference of the objective value in consecutive iterations)]{
\includegraphics[width=1.00\textwidth, height = 0.110\textheight]{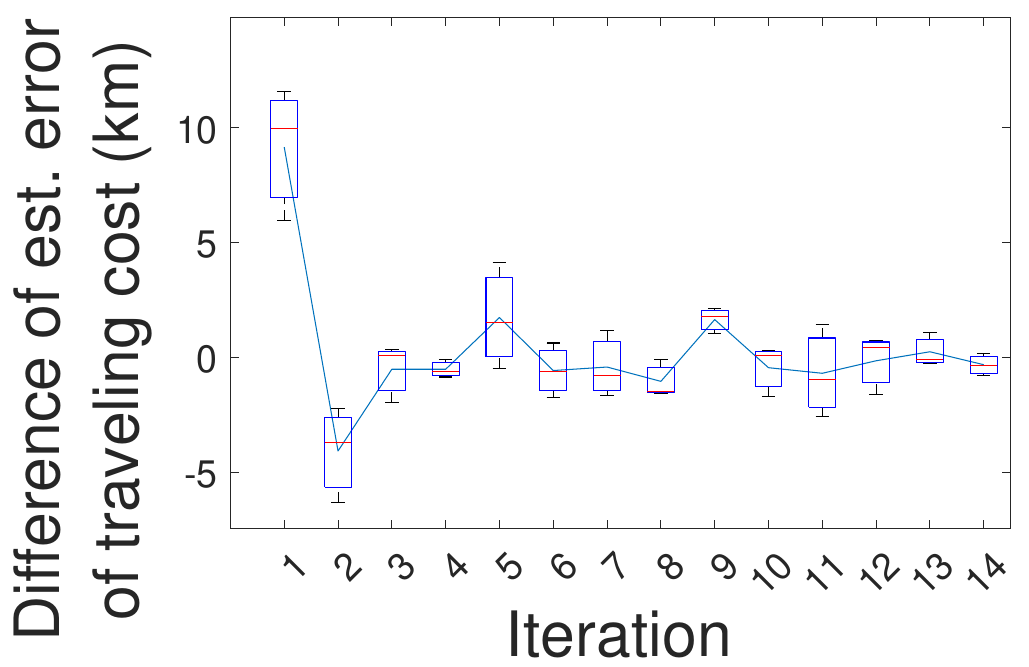}}
\label{fig:convergence3}
\end{minipage}
\hspace{0.1in}
\begin{minipage}{0.22\textwidth}
\centering
  \subfigure[\small $\delta$ = 4 (difference of the objective value in consecutive iterations)]{
\includegraphics[width=1.00\textwidth, height = 0.110\textheight]{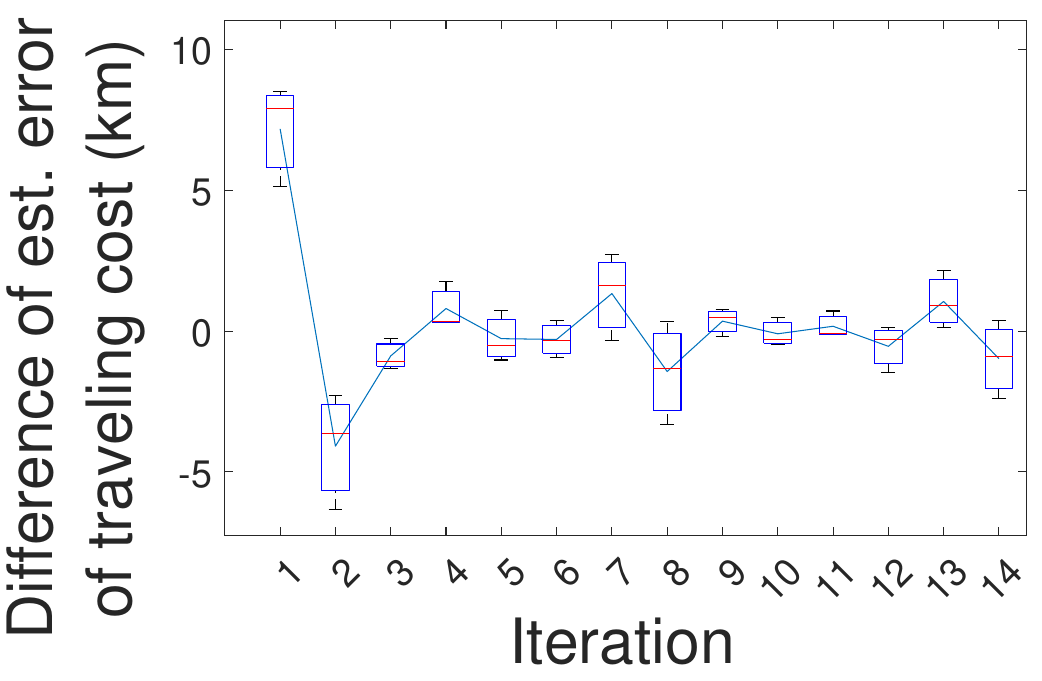}}
\label{fig:convergence4}
\end{minipage}
\vspace{-0.00in}
\caption{\normalsize Convergence of the objective value (estimation error of traveling costs).}
\label{fig:experiment1}
\vspace{-0.2in}
\end{figure}

\subsubsection{Computation time of the obfuscation matrix generation} 
\systemName{} uses graph approximation to improve the time-efficiency of the obfuscation matrix generation (Section \ref{subsec:graphapprox}). In this experiment, we evaluate how much computation time is reduced by the graph approximation. Fig. \ref{fig:experiment3a}(a) compares the computation time with and without graph approximation, with $\delta$ increased from 1 to 7. Fig. \ref{fig:experiment3a}(a) demonstrates that the graph approximation has reduced the running time by 92.34\% on average. The graph approximation improves the time efficiency of the matrix generation significantly since it reduces the number of Geo-Ind constraints from $O(K^3)$ to $O(K^2)$. Fig. \ref{fig:experiment3a}(b) compares the number of Geo-Ind constraints without and with graph approximation, with the number of locations increasing from 7 to 49. The figure shows that the number of Geo-Ind constraints is reduced by 54.58\% on average.

\begin{figure}[t]
\centering
\begin{minipage}{0.23\textwidth}
\centering
  \subfigure[\small Running time]{
\includegraphics[width=1.00\textwidth, height = 0.140\textheight]{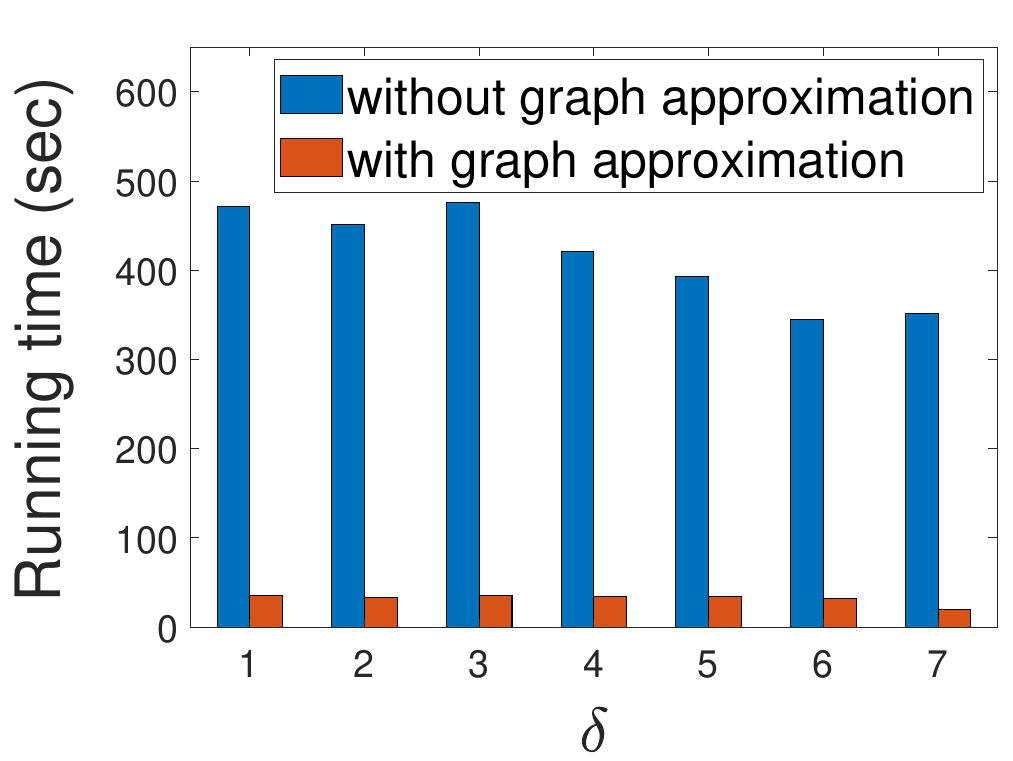}}
\end{minipage}
\hspace{-0.0in}
\begin{minipage}{0.23\textwidth}
\centering
  \subfigure[\small Number of Geo-Ind constraints]{
\includegraphics[width=1.00\textwidth, height = 0.140\textheight]{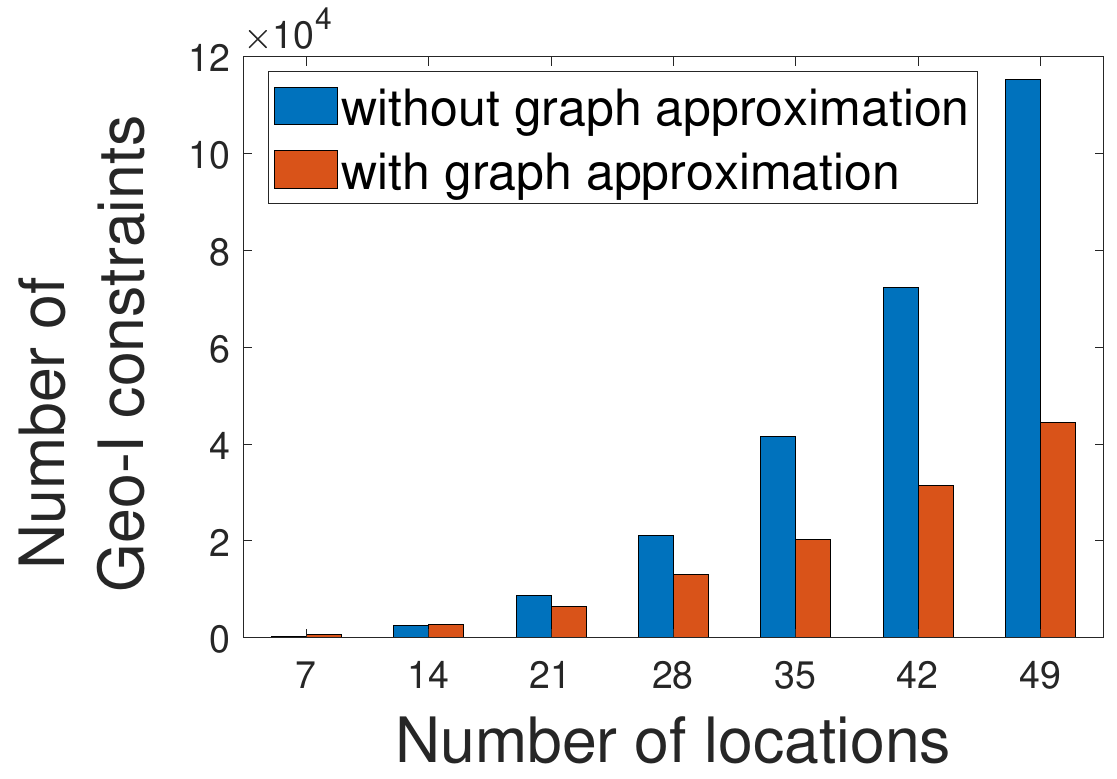}}
\end{minipage}
\vspace{-0.10in}
\caption{\normalsize Efficacy of Graph Approximation}
\label{fig:experiment3a}
\vspace{-0.10in}
\end{figure}

\subsubsection{Impact of privacy parameters}

\begin{figure}[t]
\centering
\begin{minipage}{0.30\textwidth}
\centering
  \subfigure{
\includegraphics[width=1.00\textwidth, height = 0.160\textheight]{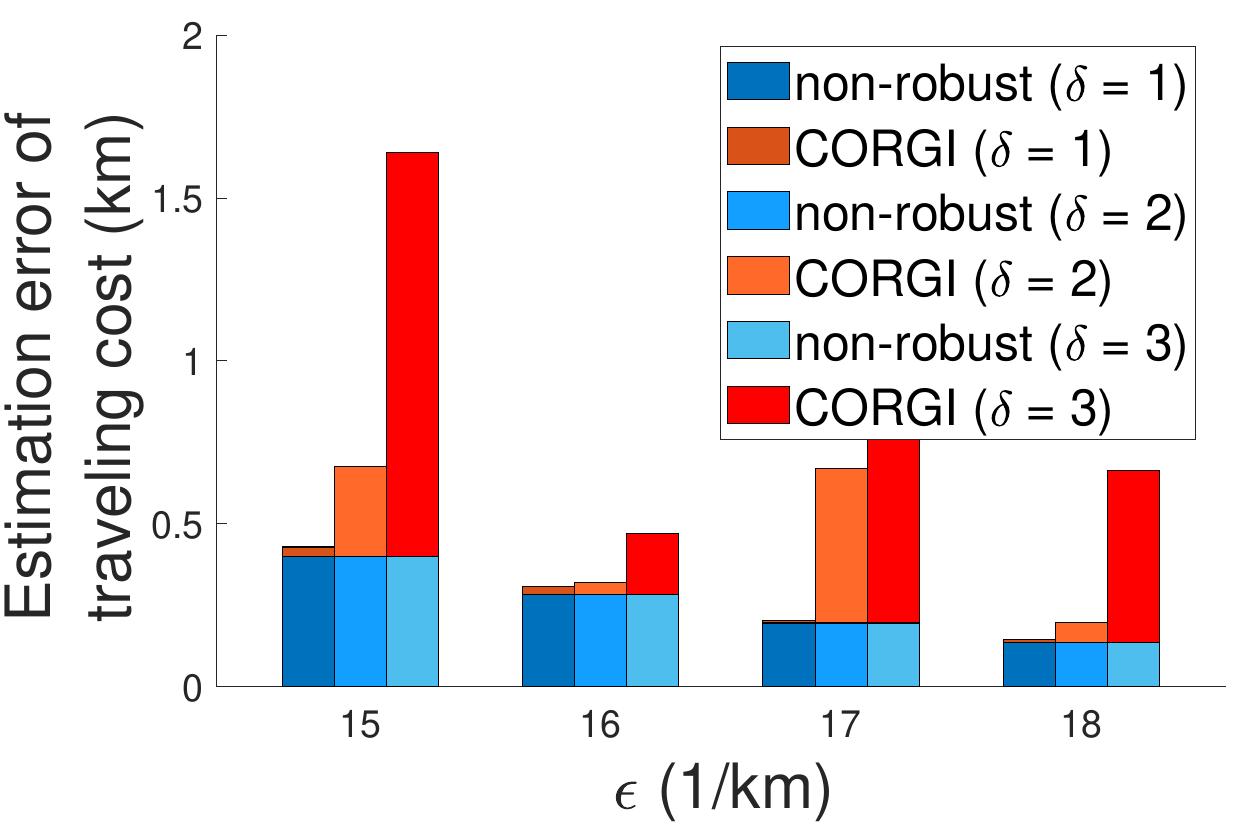}}
\end{minipage}
\vspace{-0.00in}
\caption{\normalsize Impact of privacy parameter ($\epsilon$) and customization parameter ($\delta$) on quality}
\label{fig:experiment3b}
\vspace{-0.00in}
\end{figure}

\DEL{
\begin{figure}[t]
\centering
\begin{minipage}{0.23\textwidth}
\centering
  \subfigure[\small $\epsilon$=50]{
\includegraphics[width=1.00\textwidth, height = 0.110\textheight]{./fig/exp/costdifferentdelta}}
\end{minipage}
\hspace{-0.0in}
\begin{minipage}{0.23\textwidth}
\centering
  \subfigure[\small  $\epsilon$=70]{
\includegraphics[width=1.00\textwidth, height = 0.110\textheight]{./fig/exp/costdifferentdelta_70}}
\end{minipage}
\vspace{-0.00in}
\caption{\normalsize $\delta$ vs Utility}
\label{fig:experiment3}
\vspace{-0.00in}
\end{figure}}

In this experiment, we test the impact of privacy parameter $\epsilon$ and customization parameter $\delta$ on quality loss. 
The Gowalla dataset was split into training (90\%) and testing (10\%) portions which were used for computing priors and  sampling ``real locations''  of the user. When generating the matrix, we set $NR\_TARGET$ = 49 (same as earlier) and used priors from the Gowalla dataset. We compared our results against the baseline (``non-robust'') approach which has $\delta$ = 0 and therefore is not robust against pruning of any locations from the matrix, and depict the results Fig. \ref{fig:experiment3b}(a)(b). In both Fig. \ref{fig:experiment3b}(a)(b), the $y$-axis denotes the quality loss. %
In Fig. \ref{fig:experiment3b}(a), the $x$-axis denotes the $\epsilon$ value that ranges from 15/km to 20/km in increments of 1/km. In Fig. \ref{fig:experiment3b}(b),  the $x$-axis denotes the $\delta$ value that ranges from 1 to 5. 
As illustrated in Fig.~\ref{fig:experiment3b}(a), with increasing privacy parameter $\epsilon$ the quality loss decreases, since higher $\epsilon$ implies weaker Geo-Ind constraints and hence lower quality loss.
As Fig.~\ref{fig:experiment3b}(b) shows, higher $\delta$ also introduces higher quality loss, as higher privacy budget $\epsilon'_{i,j}$ is needed for each pair of real locations $v_i$ and $v_j$ (according to Equ. (\ref{eq:approximateepsilon})). %

\vspace{-0.0in}
\subsubsection{Impact of pruning locations}

\DEL{
\begin{figure}[t]
\centering
\begin{minipage}{0.23\textwidth}
\centering
  \subfigure[\small Number of removed locations = 5]{
\includegraphics[width=1.00\textwidth, height = 0.110\textheight]{./fig/exp/violations}}
\end{minipage}
\hspace{-0.0in}
\begin{minipage}{0.23\textwidth}
\centering
  \subfigure[\small Number of removed locations = 7]{
\includegraphics[width=1.00\textwidth, height = 0.110\textheight]{./fig/exp/violations_7}}
\end{minipage}
\vspace{-0.00in}
\caption{\normalsize $\delta$ vs Privacy Violations}
\label{fig:experiment4}
\vspace{-0.00in}
\end{figure}}

\begin{figure}[t]
\centering
\begin{minipage}{0.22\textwidth}
\centering
  \subfigure[$\delta=3$]{
\includegraphics[width=1.00\textwidth, height = 0.120\textheight]{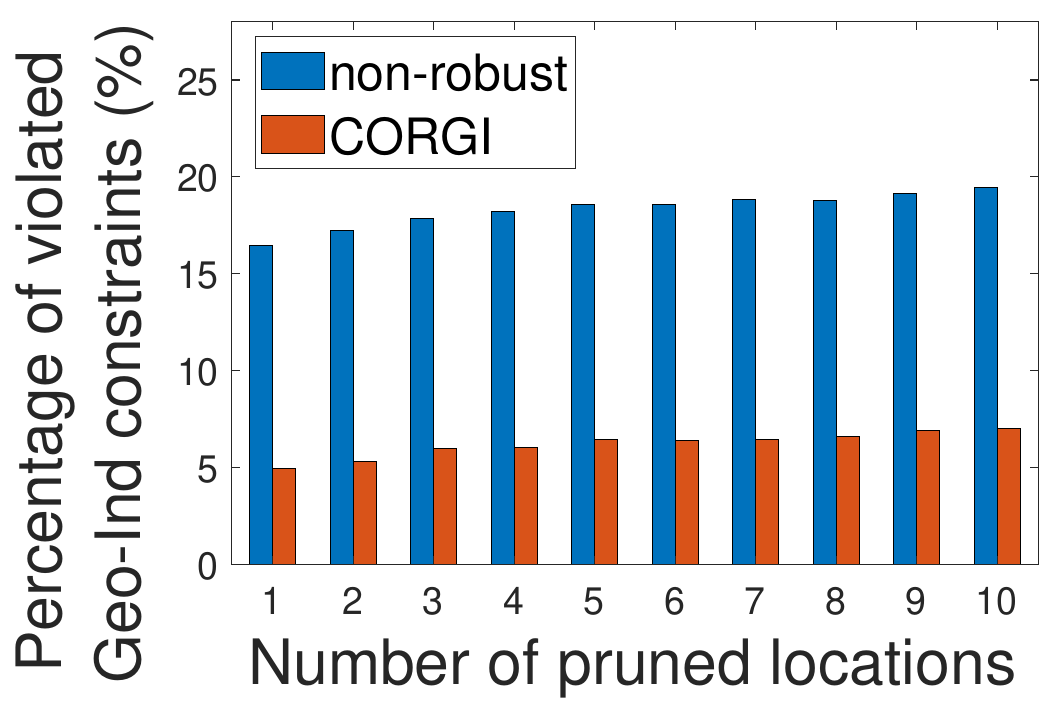}}
\end{minipage}
\hspace{0.1in}
\begin{minipage}{0.22\textwidth}
\centering
  \subfigure[$\delta=5$]{
\includegraphics[width=1.00\textwidth, height = 0.120\textheight]{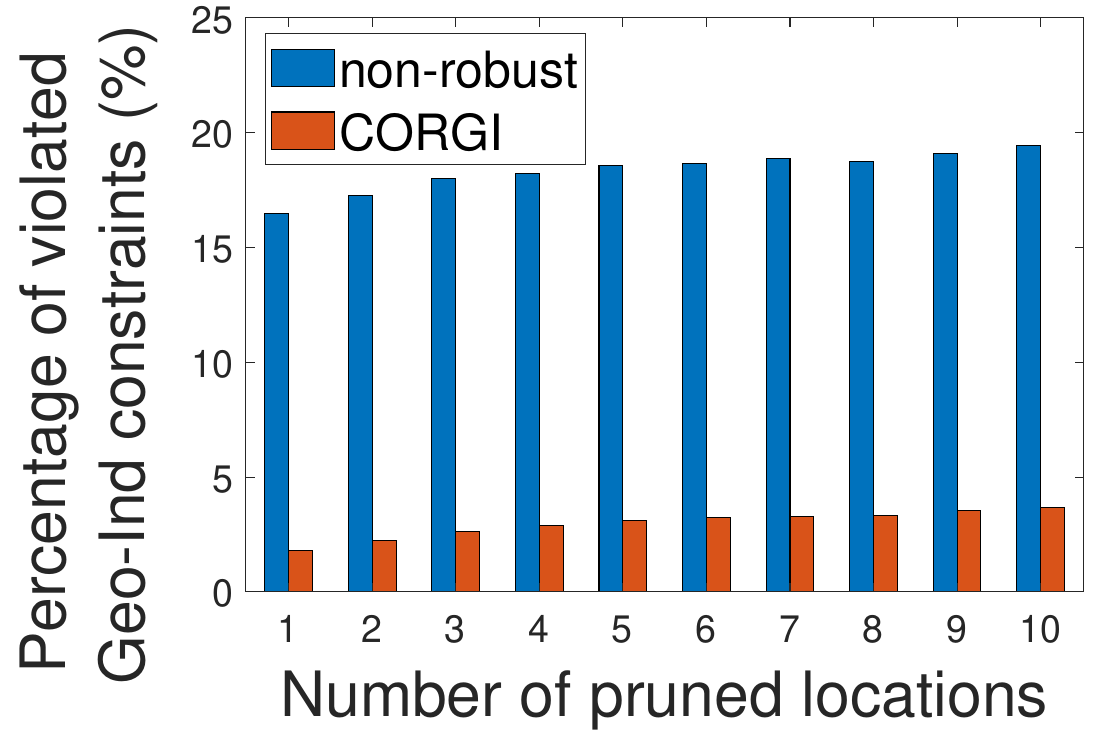}}
\end{minipage}
\vspace{-0.10in}
\caption{\normalsize Impact of customization parameter ($\delta$) on Geo-Ind violations. }
\label{fig:experiment4a}
\vspace{-0.10in}
\end{figure}

Users might not strictly follow the preferences that, only $\delta$ locations can be pruned (see Section~\ref{subsec:discussion}). Therefore, in this experiment, we test the impact of the number of locations pruned on the quality loss, especially when this number is higher than $\delta$. We create 10 experiment groups, wherein each group $n$ ($n = 1, ..., 10$), we let a user randomly prune $n$ locations from the leaf nodes of the location tree and run the experiment 500 times. We test both \systemName{} and the baseline and depict the results in Fig.  \ref{fig:experiment4a}(a)(b), where the number of locations is 49 and 70, respectively. In both figures, the $x$-axis denotes the number of locations pruned by a user, which is increased from 1 to 10, and the $y$-axis denotes the number of Geo-Ind constraint violations. 
As expected, the number of privacy violations in the non-robust matrix is much higher than that of the robust matrix. For example, pruning 14.28\% locations only causes 3.07\% Geo-Ind constraint violations in the matrix generated by \systemName{}, while it causes 18.58\% Geo-Ind constraint violations in the non-robust matrix. We also observe that with higher $\delta$, \systemName{} is more robust to the pruned locations as it preserves a higher privacy budget in the Geo-Ind constraints. The small number of privacy violations in some robust matrices is due to, 1) the number of pruned locations is greater than $\delta$ (the maximum number of locations expected to be removed) and 2) the robust matrix generation algorithm only converges to a relatively small threshold instead of 0 in consecutive iterations, indicating the output matrices might still have a small number of entries violating the preserved privacy budget.

\subsubsection{Impact of privacy level on quality loss} In this experiment, we test the quality loss of \systemName{} given different privacy levels. The location tree has four levels, where level 3 includes the root node covering 343 locations, level 2, 1, and 0 includes 49 locations, 7 locations, and 1 location, respectively. Here, we compare two possible choices from users: \textcircled{1} privacy level = 3 (with precision level = 1), and
\textcircled{2} privacy level = 2 (with precision level = 0). Fig. \ref{fig:experiment4b}(a)(b) compare the quality loss of the two choices given different $\epsilon$ and $\delta$ values. Not surprisingly, the quality loss of both choices decreases with the increase of $\epsilon$ and increases with the increase of $\delta$, which are consistent with the results in Fig. \ref{fig:experiment3b}. Furthermore, the quality loss of privacy level 3 is higher than that of privacy level 2, since level 3 has a wider range of obfuscated locations to select for users (covering 343 leaf nodes) compared to level 2 (covering 343 leaf nodes), and hence leads to a higher distortion between estimation error of traveling cost. 

\begin{figure}[t]
\centering
\begin{minipage}{0.22\textwidth}
\centering
  \subfigure[Quality loss with different $\epsilon$]{
\includegraphics[width=1.00\textwidth, height = 0.110\textheight]{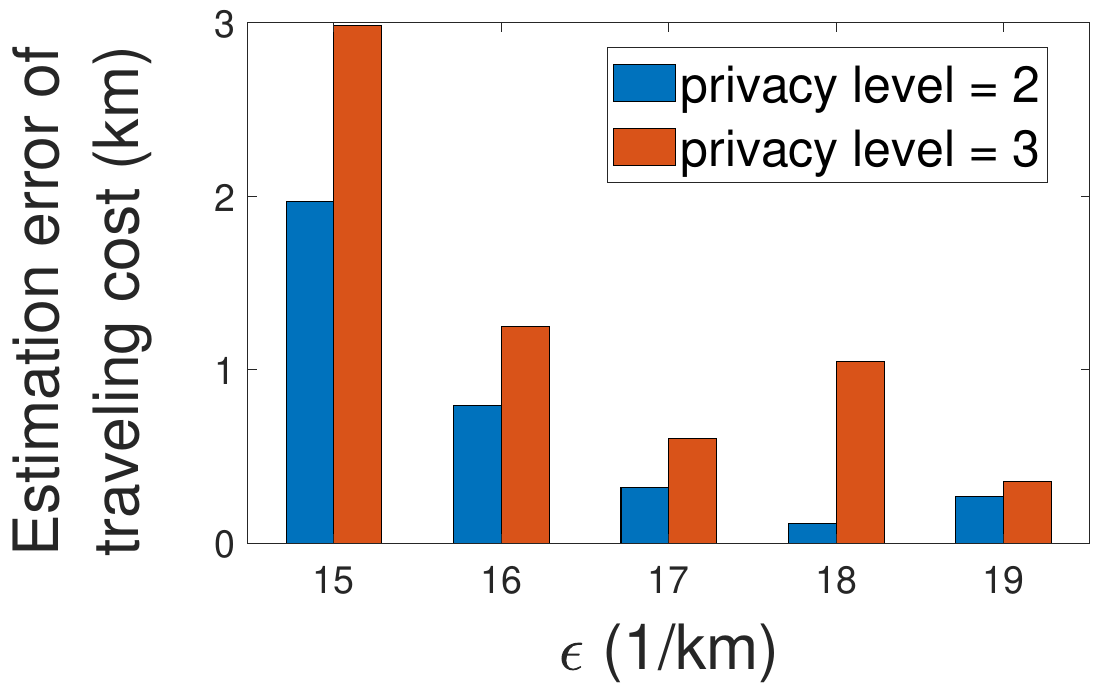}}
\end{minipage}
\hspace{0.1in}
\begin{minipage}{0.22\textwidth}
\centering
  \subfigure[Quality loss with different $\delta$]{
\includegraphics[width=1.00\textwidth, height = 0.110\textheight]{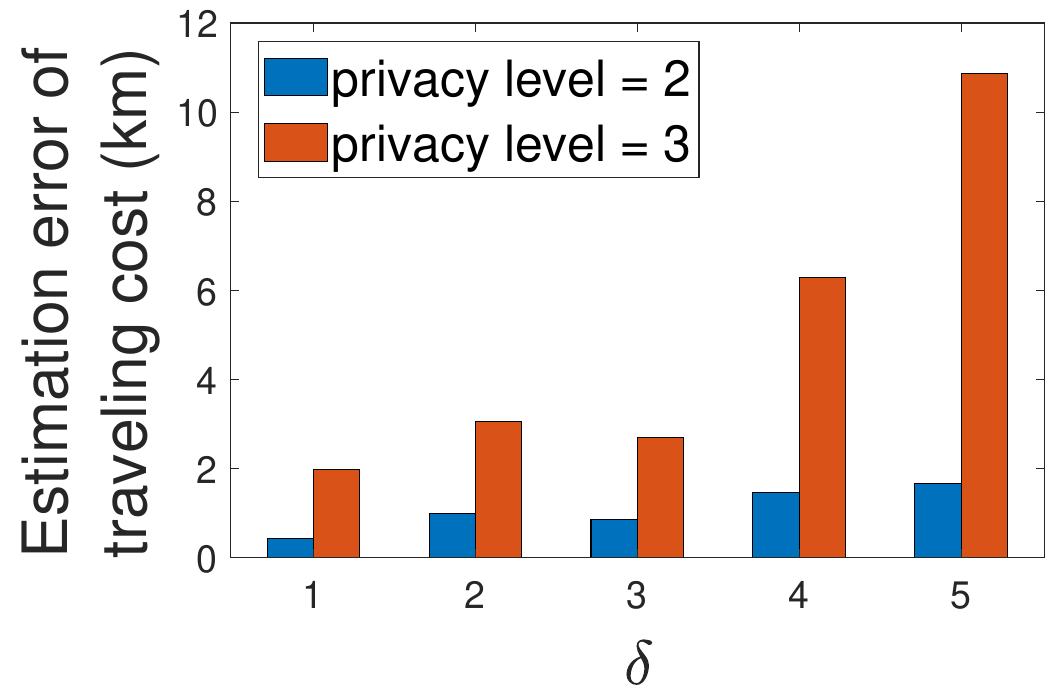}}
\end{minipage}
\vspace{-0.10in}
\caption{\normalsize Impact of obfuscation range (privacy level) on quality loss.}
\label{fig:experiment4b}
\vspace{-0.20in}
\end{figure}

\subsubsection{Computation time (precision reduction vs. matrix recalculation)} 
Recall that in \systemName{}, the server first generates the obfuscation matrix at the bottom level. If a user selects an obfuscation matrix at a higher level (which has a lower precision level), then instead of recalculating the matrix, \systemName{} generates the matrix via the precision reduction of the matrix at the bottom level. As such, in the last experiment, we test the computation time of precision reduction with the comparison of matrix recalculation. Fig. \ref{fig:experiment5}(a)(b)  shows the running time of the two approaches given the different numbers of locations (from 28 to 70) and different $\delta$ (from 1 to 7). Both figures demonstrate that precision reduction can significantly reduce the computation time compared to matrix recalculation, e.g., on average, the computation time of precision reduction is only 0.000073\% of that of the matrix recalculation. 
\begin{figure}[t]
\centering
\begin{minipage}{0.22\textwidth}
\centering
  \subfigure[Running time with different number of locations]{
\includegraphics[width=1.00\textwidth, height = 0.120\textheight]{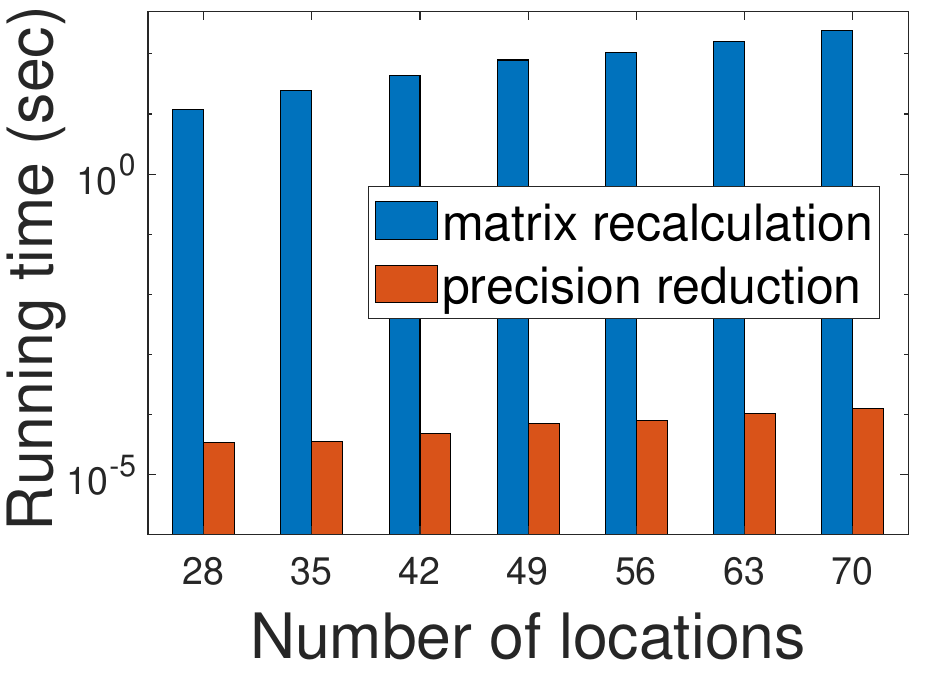}}
\end{minipage}
\hspace{0.1in}
\begin{minipage}{0.22\textwidth}
\centering
  \subfigure[Running time with different values of $\delta$]{
\includegraphics[width=1.00\textwidth, height = 0.120\textheight]{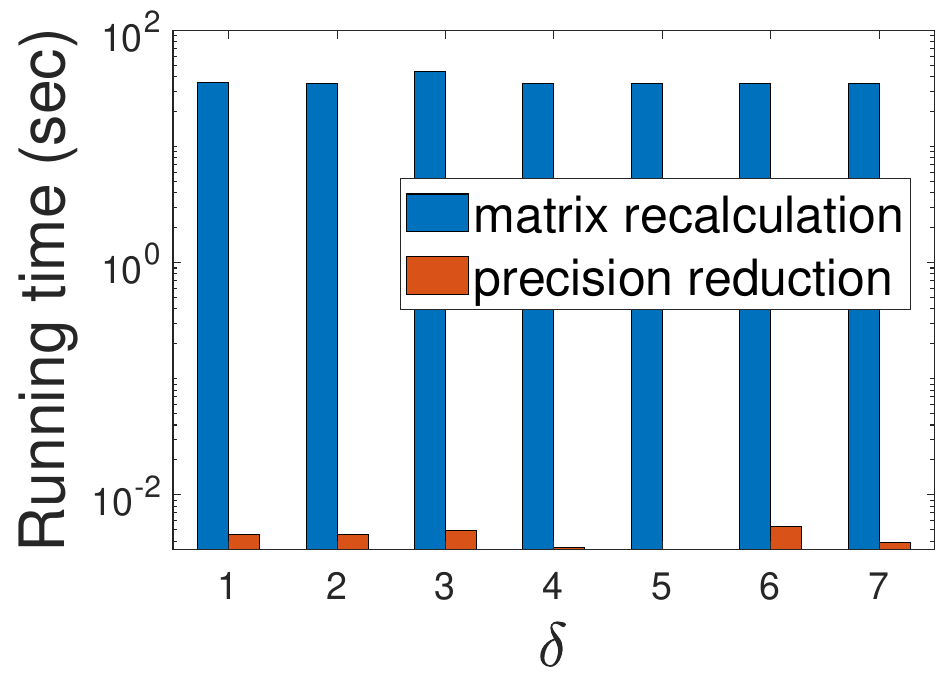}}
\end{minipage}
\vspace{-0.00in}
\caption{\normalsize Efficacy of precision reduction. }
\label{fig:experiment5}
\vspace{-0.00in}
\end{figure}

\vspace{-0.0in}
\section{Related Work}
\label{sect:related_work}

\textbf{Geo-I based obfuscation}. The discussion of location privacy criteria can date back to almost two decades ago, when Gruteser and Grunwald \cite{Gruteser-MobiSys2003} first introduced the notion of \emph{location $k$-anonymity} on the basis of Sweeney's well-known concept of \emph{$k$-anonymity} for data privacy \cite{Sweeney2002}. Location $k$-anonymity was originally used to hide a  user's identity  in LBS \cite{Zheng-IEEEAccess2018}. This notion has been extended to obfuscate location by means of \emph{$l$-diversity}, i.e., a user's location cannot be distinguished with other $l-1$ locations \cite{Wang-VLDB2009,Yu-NDSS2017}.
However, $l$-diversity is hard to achieve in many applications as it assumes dummy locations are equally likely to be the real location from the attacker's view \cite{Andres-CCS2013,Yu-NDSS2017}.  \looseness = -1

In recent years, the privacy notion \emph{Geo-Ind}
\cite{Andres-CCS2013} which was first by introduced by Andres et. al, and many obfuscation strategies based on it (e.g., \cite{Shokri-SP2011,Shokri-CCS2012,Andres-CCS2013,Wang-WWW2017,Bordenabe-CCS2014,Ardagna-TDSC2011,Shokri-SP2011, Yu-NDSS2017, Qiu-CIKM2020}) have been used for location obfuscation. As these strategies inevitably introduce errors to users' reported locations, leading to a quality loss in LBS, a key issue that has been discussed in those works is how to trade off QoS and privacy. Many existing works follow a global optimization framework: given the Geo-I constraints, an optimization model is formulated to minimize the quality loss caused by obfuscation  \cite{Shokri-CCS2012,Fawaz-Security2015,Fawaz-CCS2014,Wang-WWW2017}. 
We now cover the related work closer to our work by categorizing them into tree-based approaches to obfuscation and policy-based approaches to customization.

\vspace{0.05in}
\noindent \textbf{Tree/hierarchy based approaches to location obfuscation}. 
\cite{ahuja2019utility} uses a hierarchical grid to overcome the computational overhead of optimal mechanisms. They first construct a hierarchical grid with increasing granularity as one traverses down the index with the highest granularity at leaf nodes (similar to our approach).  Second, they allocate the privacy budget ($\epsilon$) appropriately to these different levels using sequential composition.
In order to generate the obfuscated location, they start at the root node containing the real location of the user and go down the tree by recursively using the output of the obfuscation function at the prior level.
The main difference between our approach and theirs is, they partition the privacy budget for each level in the grid, while ours, no matter from top to bottom or bottom to up (increase or decrease precision), uses the maximum privacy budget. 
In \cite{tao2020differentially}, the authors present a tree-based approach for differentially private online task assignments for crowdsourcing applications.
They construct a Hierarchically well-Separated Tree (HST) based on a region that is published to both workers and task publishers who use it in order to obfuscate worker and task locations respectively.
They show that this tree-based approach achieves $\epsilon$\textit{-Geo-Ind} while minimizing total distance (maximizing utility) for task assignments.
However, their approach relies on workers and task publishers using the same HST and obfuscation function in order to effectively perform task assignments and is not geared toward allowing users to customize the obfuscation functions.
Other hierarchical-based approaches to spatial data such as \cite{cormode2012differentially, 10.1145/3474717.3483943} focus on private release population statistics or histograms.

\vspace{0.05in}
\noindent \textbf{Policy based approach to privacy}. Blowfish privacy proposed by \cite{he2014blowfish} uses a policy graph to determine the set of neighbors that users want to mark as sensitive. A policy graph encodes the user's preferences about which pairs of values in the domain of the database should be indistinguishable for an adversary. Thus, it allows users to tradeoff privacy for utility by restricting the indistinguishability set. Blowfish works for statistical queries and not location queries.
\cite{cao2020pglp} extended blowfish privacy and applied it to location privacy where the nodes and edges in the policy graph represent possible locations of the user and the indistinguishability requirements respectively.
Their goal is to ensure $\epsilon$-\textit{Geo-Ind} for any two connected nodes in the graph and to achieve this they apply DP-based noise to latitude and longitude independently. 
Their approach is best suited for category-based privacy i.e., indistinguishability among multiple locations of the same category (e.g., restaurants) as specifying pairwise indistinguishability between locations according to general user preferences is challenging.
Our customization model allows users to specify their preferences which then get translated to the parameters in generating the obfuscation function reducing the overhead on the user in terms of specification.
Furthermore, \cite{cao2020pglp} does not allow users to choose the granularity at which their location is shared as the graph model doesn't capture the natural hierarchy of locations. \cite{asada2019hide} proposed an approach to recommend location privacy preferences based on place and time (similar to our customization policies) using local differential privacy. Their work is complementary to ours and could be used to help users in coming up with their user preferences.

\vspace{-0.0in}
\section{Conclusions}
\label{sect:conclusions}

We developed \systemName{}, a framework for generating customizable obfuscation functions with strong privacy guarantees via Geo-Indistinguishability. \systemName{} includes a location tree and a policy model to assist users in specifying their customization parameters. \systemName{} includes user and server side interactions for efficiently generating a robust matrix. Experimental results show that \systemName{} effectively balances privacy, utility, and customization.

\balance
\bibliographystyle{IEEEtran}
\bibliography{references.bib}

\end{document}